\documentclass[UKenglish,a4paper,pdfa,cleveref]{lipics-v2021}

\usepackage{booktabs,csquotes,etoolbox,mathdots,mathtools,tikz,xparse}
\usetikzlibrary{babel,cd}

\title{Computing the Fr\'echet Distance Between Uncertain Curves in One Dimension}
\author{Kevin Buchin}{Department of Mathematics and Computer Science, TU Eindhoven, Netherlands \and
\url{https://www.win.tue.nl/~kbuchin/}}{k.a.buchin@tue.nl}{https://orcid.org/0000-0002-3022-7877}{}
\author{Maarten L\"offler}{Department of Information and Computing Sciences, Utrecht University,
Netherlands \and\url{https://webspace.science.uu.nl/~loffl001/}}{m.loffler@uu.nl}{}{Partially
supported by the Dutch Research Council (NWO) under project no.\@ 614.001.504.}
\author{Tim Ophelders}{Department of Information and Computing Sciences, Utrecht University,
Netherlands \and Department of Mathematics and Computer Science, TU Eindhoven, Netherlands
\and\url{https://www.win.tue.nl/~tophelde/}}{t.a.e.ophelders@uu.nl}{}{}
\author{Aleksandr Popov}{Department of Mathematics and Computer Science, TU Eindhoven, Netherlands
\and\url{https://www.win.tue.nl/~apopov/}}{a.popov@tue.nl}{https://orcid.org/0000-0002-0158-1746}{Supported
by the Dutch Research Council (NWO) under project no.\@ 612.001.801.}
\author{J\'er\^ome Urhausen}{Department of Information and Computing Sciences, Utrecht University,
Netherlands}{j.e.urhausen@uu.nl}{}{}
\author{Kevin Verbeek}{Department of Mathematics and Computer Science, TU Eindhoven, Netherlands
\and\url{https://www.win.tue.nl/~kverbeek/}}{k.a.b.verbeek@tue.nl}{}{}

\authorrunning{K.~Buchin, M.~L\"offler, T.~Ophelders, A.~Popov, J.~Urhausen, and K.~Verbeek}
\Copyright{Kevin Buchin, Maarten L\"offler, Tim Ophelders, Aleksandr Popov, J\'er\^ome Urhausen, and
Kevin Verbeek}
\acknowledgements{Research on the topic of this paper was initiated at the 5th Workshop on Applied
Geometric Algorithms (AGA 2020) in Langbroek, Netherlands.}

\hideLIPIcs
\nolinenumbers

\graphicspath{{../}}

\theoremstyle{definition}
\newtheorem{problem}[theorem]{Problem}

\newcommand*{\dsh}{~--~}
\newcommand*{\fr}{d_\mathrm{F}}
\newcommand*{\dfr}{d_\mathrm{dF}}
\newcommand*{\wfr}{d_\mathrm{wF}}
\newcommand*{\frmin}{\fr^{\,\min}}
\newcommand*{\wfrmin}{\wfr^{\,\min}}
\newcommand*{\dfrmax}{\dfr^{\,\max}}
\newcommand*{\frmax}{\fr^{\,\max}}
\newcommand*{\U}{\mathcal{U}}
\newcommand*{\V}{\mathcal{V}}
\newcommand*{\R}{\mathbb{R}}
\newcommand*{\Reg}{\mathcal{R}}
\newcommand*{\Fd}{\mathcal{F}_{\delta}}
\newcommand*{\dir}{\mathbf{d}}

\newcommand{\pr}{\mathrm{Pr}}
\newcommand{\Ima}{\mathrm{Im}}
\newcommand*{\True}{\textsf{True}}
\newcommand*{\False}{\textsf{False}}
\newcommand*{\concat}{\sqcup}
\newcommand*{\Concat}{\bigsqcup}
\newcommand{\rMono}{\mathit{rm}}

\newcommand{\cost}{\mathrm{cost}}
\newcommand{\grow}[1]{\overrightarrow{#1}}
\newcommand{\rev}[1]{#1^{\raisebox{-2pt}{\(\scriptscriptstyle -1\)}}}
\newcommand*{\ia}[2]{\mathmakebox[\widthof{\(#2\)}]{#1}}

\definecolor{myBlue}{rgb}{0.121, 0.47, 0.705}
\definecolor{myGreen}{rgb}{0.2, 0.627, 0.172}
\definecolor{myYellow}{rgb}{1.000, 0.604, 0.000}
\definecolor{myRed}{rgb}{0.89, 0.102, 0.109}
\definecolor{myPurple}{rgb}{0.415, 0.239, 0.603}

\begin{document}
\maketitle

\begin{abstract}
We consider the problem of computing the Fr\'echet distance between two curves for which the exact
locations of the vertices are unknown.
Each vertex may be placed in a given \emph{uncertainty region} for that vertex, and the objective is
to place vertices so as to minimise the Fr\'echet distance.
This problem was recently shown to be NP-hard in~2D, and it is unclear how to compute an optimal
vertex placement at all.

We present the first general algorithmic framework for this problem.
We prove that it results in a polynomial-time algorithm for curves in~1D with intervals as
uncertainty regions.
In contrast, we show that the problem is NP-hard in~1D in the case that vertices are placed to
maximise the Fr\'echet distance.

We also study the weak Fr\'echet distance between uncertain curves.
While finding the optimal placement of vertices seems more difficult than the regular Fr\'echet
distance\dsh and indeed we can easily prove that the problem is NP-hard in~2D\dsh the optimal
placement of vertices in~1D can be computed in polynomial time.
Finally, we investigate the discrete weak Fr\'echet distance, for which, somewhat surprisingly, the
problem is NP-hard already in~1D.
\end{abstract}

\section{Introduction}
The \emph{Fr\'echet distance} is a popular distance measure for curves.
Its computational complexity has drawn considerable attention in computational
geometry~\cite{agarwal:2014,alt:1995,bringmann:2014,bringmann:2016,buchin:2017,driemel:2012,harpeled:2014}.
The Fr\'echet distance between two (polygonal) curves is often illustrated using a person and a dog:
imagine a person is walking along one curve having the dog, which walks on the other curve, on a
leash.
The person and the dog may change their speed independently but may not walk backwards.
The Fr\'echet distance corresponds to the minimum leash length needed with which the person and the
dog can walk from start to end on their respective curve.

The Fr\'echet distance and its variants have found many applications, for instance, in the context
of protein alignment~\cite{jiang:2008}, handwriting recognition~\cite{zheng:2008}, map
matching~\cite{brakatsoulas:2005} and construction~\cite{ahmed:2012,buchin-mapconstruction:2017}, and
trajectory similarity and clustering~\cite{buchin-klcluster:2019,gudmundsson:2014}.
In most of these applications, we obtain the curves by a sequence of measurements, and these
measurements are inherently imprecise.
However, it is often reasonable to assume that the true location is within a certain radius of the
measurement, or more generally that it stays within an \emph{uncertainty region}.

Re-imagine the person and the dog, except now each is given a sequence of regions they have to visit.
More specifically, they need to visit one location per region and move on a straight line between
locations without going backwards.
Suppose they need to minimise the leash length.
This corresponds to the following problem.
Each curve is given by a sequence of uncertainty regions; we minimise the Fr\'echet distance over
all possible choices of locations in the regions.
This is called the \emph{lower bound} problem for the Fr\'echet distance between uncertain curves.

Similar problems involving uncertainty have drawn more and more attention in the past few years in
computational geometry.
Most results are on uncertain point sets, where we often aim to minimise or maximise some quantity
stemming from the point set, but also perform visibility queries in polygons or find Delaunay
triangulations~\cite{abellanas:2001,buchin:2011,fan:2013,knauer:2011,loeffler:2009,loeffler:2014,
loeffler:2010,loeffler:2006,kreveld:2010}.
More recently there have also been several results on curves with
uncertainty~\cite{ahn:2012,icalp,sijben_dfd,fan:2018}.

The earliest results for a variant of the problem we consider do not concern the Fr\'echet distance as
such, but its variant the discrete Fr\'echet distance, where we restrict our attention to the
vertices of the curves.
Ahn et al.~\cite{ahn:2012} show a polynomial-time algorithm that decides whether the lower bound
discrete Fr\'echet distance is below a certain threshold, for two curves with uncertainty regions
modelled as circles in constant dimension.
The lower bound Fr\'echet distance with uncertainty regions modelled as point sets admits a simple
dynamic program~\cite{icalp}.
However, as has been recently shown, the decision problem for the continuous Fr\'echet distance is
NP-hard already in two dimensions with vertical line segments as uncertainty regions and one precise
and one uncertain curve~\cite{icalp}; it is not even clear how to compute the lower bound at all
with any uncertainty model that is not discrete.
We present a general algorithmic framework for computing the lower bound Fr\'echet distance that can
be instantiated in many settings.
In the general 2D~setting, this gives an exponential-time algorithm; we turn our attention to curves
in~1D.
We instantiate our framework in~1D and show that it results in an efficient algorithm for
imprecision modelled as intervals.

Next to the discrete Fr\'echet distance, the most common variant of the Fr\'echet distance is the
weak Fr\'echet distance~\cite{alt:1995}.
In the person--dog analogy, this variant allows backtracking on the paths.
The weak Fr\'echet distance (for certain curves) has interesting properties
in~1D~\cite{dog_difficult,buchin:2019,harpeled:2014}: it can be computed in linear time in~1D, while
in~2D it cannot be computed significantly faster than quadratic time under the strong
exponential-time hypothesis.
To our knowledge, the weak Fr\'echet distance has not been studied in the uncertain setting before.
We give a polynomial-time algorithm that solves the lower bound problem in~1D.
In contrast to that, we show that the problem is NP-hard in~2D, and that discrete weak Fr\'echet
distance is NP-hard already in~1D.
We summarise these results in \cref{tab:complexity}.

\begin{table}
\centering
\caption{Complexity results for the lower bound problems for uncertain curves.}
\label{tab:complexity}
\begin{tabular}{l c c c c}
\toprule
   &\multicolumn{2}{c}{Fr\'echet distance}             & \multicolumn{2}{c}{Weak Fr\'echet distance}\\
   & discrete                   & continuous           & discrete & continuous\\
\midrule
1D & polynomial~\cite{ahn:2012} & polynomial           & NP-hard  & polynomial\\
2D & polynomial~\cite{ahn:2012} & NP-hard~\cite{icalp} & NP-hard  & NP-hard\\
\bottomrule
\end{tabular}
\end{table}

The table provides an interesting insight.
First of all, it appears that for continuous distances the dimension matters, whereas for the
discrete ones the results are the same both in~1D and~2D.
Moreover, it may be surprising that discretising the problem has a different effect: for the
Fr\'echet distance it makes it easier, while for the weak Fr\'echet distance the problem becomes
harder.
We discuss the polynomial-time algorithm for Fr\'echet distance in~1D in \cref{sec:lb1d}.
We give the algorithm for weak Fr\'echet distance in~1D in \cref{sec:w1dalg}.
We show the NP-hardness constructions for the weak (discrete) Fr\'echet distance in
\cref{sec:wdhard}.

Finally, we also turn our attention to the problem of maximising the Fr\'echet distance, or finding
the upper bound.
It has been shown that the problem is NP-hard in~2D for several uncertainty models, including
discrete point sets, both for discrete and continuous Fr\'echet distance~\cite{icalp}.
We strengthen that result by presenting a similar construction that already shows NP-hardness in~1D.
The proof is given in \cref{sec:ubfr}.

\section{Preliminaries}
Denote \([n] \equiv \{1, 2, \dots, n\}\).
Consider a \emph{sequence} of points \(\pi = \langle p_1, p_2, \dots, p_n\rangle\).
We also use \(\pi\) to denote a \emph{polygonal curve,} defined by the sequence by linearly
interpolating between the points and can be seen as a continuous function:
\(\pi(i + \alpha) = (1 - \alpha) p_i + \alpha p_{i + 1}\) for \(i \in [n - 1]\) and
\(\alpha \in [0, 1]\).
The \emph{length} of such a curve is the number of its vertices, \(\lvert \pi\rvert = n\).
Denote the concatenation of two sequences \(\pi\) and \(\sigma\) of lengths \(n\) and \(m\)
by \(\pi \concat \sigma\): the result consists of \(\pi\), then \(\sigma\).
We can generalise this notation:
\[\pi \equiv \Concat_{i \in [n]} p_i = p_1 \concat p_2 \concat \dots \concat p_n\,.\]
Denote a subcurve from vertex \(i\) to \(j\) of \(\pi\) as
\(\pi[i: j] = p_i \concat p_{i + 1} \concat \dots \concat p_j\).
Occasionally we use the notation \(\langle\pi(i) \mid i \in I\rangle_{i = 1}^m\) to denote a curve
built on a subsequence of vertices of \(\pi\), where vertices are only taken if they are in set
\(I\).
For example, setting \(I = \{1, 3, 4\}\), \(m = 5\), \(\pi = \langle p_1, p_2, \dots, p_5\rangle\)
means \(\langle\pi(i) \mid i \in I\rangle_{i = 1}^m = \langle p_1, p_3, p_4\rangle\).

Denote the Fr\'echet distance between two polygonal curves \(\pi\) and \(\sigma\) by
\(\fr(\pi, \sigma)\), the discrete Fr\'echet distance by \(\dfr(\pi, \sigma)\), and the
weak Fr\'echet distance by \(\wfr(\pi, \sigma)\).
Recall the definition of Fr\'echet distance for polygonal curves of lengths \(m\) and \(n\).
It is based on \emph{parametrisations} (non-decreasing surjections) \(\alpha\) and \(\beta\) with
\(\alpha \colon [0, 1] \to [1, m]\), \(\beta \colon [0, 1] \to [1, n]\).
Parametrisations establish a \emph{matching}.
Denote the cost of a matching \(\mu = (\alpha, \beta)\) as
\(\cost_\mu(\pi, \sigma) = \max_{t \in [0, 1]} \lVert\pi\circ\alpha(t) - \sigma\circ\beta(t)\rVert\).
Then we can define Fr\'echet distance and its variants as
\[\fr(\pi, \sigma) = \inf_{\text{matching } \mu} \cost_\mu(\pi, \sigma)\,,\quad
\fr(\pi, \sigma) = \inf_{\text{discrete matching } \mu} \cost_\mu(\pi, \sigma)\,,\]
\[\wfr(\pi, \sigma) = \inf_{\text{weak matching }\mu} \cost_\mu(\pi, \sigma)\,.\]
The discrete matching is restricted to vertices, and the weak matching is not a pair of
parametrisations, but a path \((\alpha, \beta) \colon [0, 1]^2 \to [1, m] \times [1, n]\), with
\(\alpha(0) = 1, \alpha(1) = m\) and \(\beta(0) = 1, \beta(1) = n\).
In the person--dog analogy for the Fr\'echet distance, the best choice of parametrisations means
that the person and the dog choose the best speed, and the leash length is then the largest needed
leash length during the walk.

An \emph{uncertain} point in one dimension is a set of real numbers \(u \subseteq \R\).
The intuition is that only one point from this set represents the true location of the point;
however, we do not know which one.
A \emph{realisation} \(p\) of such a point is one of the points from \(u\).
In this paper we consider two special cases of uncertain points.
An \emph{indecisive} point is a finite set of numbers \(u =
\{x_1, \dots, x_\ell\}\).
An \emph{imprecise} point is a closed interval \(u = [x_1, x_2]\).
Note that a precise point is a special case of both indecisive and imprecise points.

Define an \emph{uncertain curve} as a sequence of uncertain points
\(\U = \langle u_1, \dots, u_n\rangle\).
A \emph{realisation} \(\pi \Subset \U\) of an uncertain curve is a polygonal curve
\(\pi = \langle p_1, \dots, p_n\rangle\), where each \(p_i\) is a realisation of the uncertain
point \(u_i\).
For uncertain curves \(\U\) and \(\V\), define the lower bound and upper bound Fr\'echet
distance.
The discrete and weak Fr\'echet distance are defined similarly.
\[\frmin(\U, \V) = \min_{\pi \Subset \U, \sigma \Subset \V} \fr(\pi, \sigma)\,,\qquad
\frmax(\U, \V) = \max_{\pi \Subset \U, \sigma \Subset \V} \fr(\pi, \sigma)\,.\]

\section{Lower Bound Fr\'echet Distance: General Approach}\label{sec:lbdd}
In this section, we consider the following decision problem.
\begin{problem}
Given two uncertain curves \(\U = \langle u_1, \dots, u_m\rangle\) and
\(\V = \langle v_1, \dots, v_n\rangle\) in \(Y = \R^d\) for some \(d, m, n \in \mathbb{N}^+\) and a
threshold \(\delta > 0\), decide if \(\frmin(\U, \V) \leq \delta\).
\end{problem}
Note that this problem formulation is general both in terms of the shape of uncertainty regions and
the dimension of the problem.
We propose an algorithmic framework that solves this problem.
As has been shown previously~\cite{icalp}, the problem is NP-hard in~2D for vertical line segments
as uncertainty regions, but admits a simple dynamic program for indecisive points in~2D.
So, in many uncertainty models, especially in higher dimensions, the following approach will not
result in an efficient algorithm.
However, our approach is general in that it can be instantiated in restricted settings, e.g.\@ in~2D
assuming that the segments of the curves can only be horizontal or vertical.
The inherent complexity of the problem appears to be related to the number of directions to
consider, with the infinite number in~2D without restrictions and two directions in~1D).
We conjecture that in this restricted setting the approach yields a polynomial-time algorithm;
verifying this and making a more general statement delineating the hardness of restricted settings
are both interesting open problems.
Our approach shows a straightforward way to engineer an algorithm for various restricted settings
in arbitrary dimension, but we cannot make any statements about its efficiency in most settings.
To illustrate the approach, we instantiate it in~1D and analyse its efficiency in \cref{sec:lb1d}.
The interested reader might refer to that section for a more intuitive explanation of the approach.

First we introduce some extra notation.
For \(i \in [m]\), denote \(\U_i = \langle u_1, \dots, u_i\rangle\) and
\(\U_i^* = \langle u_1, \dots, u_i, Y\rangle\).
We call \(\U_i\) and \(\U_i^*\) the \emph{subcurve} and the \emph{free subcurve} of \(\U\) at \(i\),
respectively.
Intuitively, a realisation of \(\U_i^*\) extends a realisation of \(\U_i\) by a single edge whose
final vertex position is unrestricted.
Let \(S := S^{d - 1}\) be the unit \((d - 1)\)-sphere.
Denote the \emph{direction} of the \(i\)-th edge \(\pi[i, i + 1]\) of a realisation \(\pi\) by
\(\dir_i(\pi) \subseteq S\).
For example, in~1D there are only two options, in~2D the directions can be picked from a unit
circle, in~3D from a unit sphere, etc.
In the degenerate case where the edge has length \(0\) (or \(\pi\) has no \(i\)-th edge), let
\(\dir_i(\pi) = S\).

We want to find realisations \(\pi \Subset \U\) and \(\sigma \Subset \V\) such that \(\pi\) and
\(\sigma\) have Fr\'echet distance at most \(\delta\).
Call such a pair \((\pi, \sigma)\) a \(\delta\)-realisation of \((\U, \V)\).
Recall that two polygonal curves \(\pi \colon [1, i] \to Y\) and \(\sigma \colon [1, j] \to Y\) have
Fr\'echet distance \(\fr(\pi, \sigma)\) at most \(\delta\) if and only if there exist
\emph{parametrisations} (non-decreasing surjections) \(\alpha \colon [0, 1] \to [1, i]\) and
\(\beta \colon [0, 1] \to [1, j]\) such that the path \((\pi \circ \alpha, \sigma \circ \beta)\)
lies in the \(\delta\)-free space
\(\Fd = \{(p, q) \in Y \times Y \mid \lVert p - q\rVert \leq \delta\}\).
For \(\delta\)-close (free) subcurves of \(\U\) at \(i\) and \(\V\) at \(j\), we capture their pairs
of endpoints and final directions using \(\Reg_{i, j}, \Reg_{i, j^*}, \Reg_{i^*, j},
\Reg_{i^*, j^*} \subseteq Y \times Y \times S \times S\):
\begin{align*}
\Reg_{i, j} &= \{(\pi(i)\hphantom{{}\mathbin{+} 1}, \sigma(j)\hphantom{{}\mathbin{+} 1}, s, t)
	\mid \pi \Subset \ia{\U_i}{\U_i^*}, \sigma \Subset \ia{\V_j}{\V_j^*},
    s \in \dir_i(\pi), t \in \dir_j(\sigma), \fr(\pi, \sigma) \leq \delta\},\\
\Reg_{i, j^*} &= \{(\pi(i)\hphantom{{}\mathbin{+} 1}, \sigma(j + 1), s, t)
	\mid \pi \Subset \ia{\U_i}{\U_i^*}, \sigma \Subset \V_j^*,
    s \in \dir_i(\pi), t \in \dir_j(\sigma), \fr(\pi, \sigma) \leq \delta\},\\
\Reg_{i^*, j} &= \{(\pi(i + 1), \sigma(j)\hphantom{{}\mathbin{+} 1}, s, t)
	\mid \pi \Subset \U_i^*, \sigma \Subset \ia{\V_j}{\V_j^*},
    s \in \dir_i(\pi), t \in \dir_j(\sigma), \fr(\pi, \sigma) \leq \delta\},\\
\Reg_{i^*, j^*} &= \{(\pi(i + 1), \sigma(j + 1), s, t)
	\mid \pi \Subset \U_i^*, \sigma \Subset \V_j^*,
    s \in \dir_i(\pi), t \in \dir_j(\sigma), \fr(\pi, \sigma) \leq \delta\}.
\end{align*}
Note that for \(\pi \Subset \U_{i}\), \(i\) is the final vertex, so \(\dir_i(\pi) = S\).
Therefore, \(\Reg_{i, j}\) captures the reachable subset of \(Y \times Y\) for the realisations of
the last points of the prefixes, and the two other dimensions contain all points from \(S\) to
capture that we may proceed in any allowed direction.
The set \(\Reg_{i^*, j}\) captures the reachable subset of \(Y \times Y\) for the point in the
parametrisation where we are between vertices \(i\) and \(i + 1\) on \(\U\) and at \(j\) on \(\V\);
we have not restricted the range to \(u_{i + 1}\) yet.
The allowed directions for parameter \(s\) now depend on how we reached this point in the
parametrisation, since segments connecting realisations are straight line segments, and the
direction needs to be kept consistent once chosen.
From this description the reader can deduce what the other sets capture by symmetry.
See also \cref{fig:cells}, where the sets are positioned as in a regular free-space diagram,
replacing the edges, vertices, and cells.

\begin{figure}
\begin{minipage}[t]{.5\linewidth}
\centering
\begin{tikzpicture}[scale=1.6]
\draw[help lines] (1.5, 1.5) grid[step=2] (4.5, 4.5);
\node[fill=white] at (2, 2) {\(\Reg_{i, j}\)};
\node[fill=white] at (4, 2) {\(\Reg_{i + 1, j}\)};
\node[fill=white] at (2, 4) {\(\Reg_{i, j + 1}\)};
\node[fill=white] at (4, 4) {\(\Reg_{i + 1, j + 1}\)};
\node[fill=white] at (3, 2) {\(\Reg_{i^*, j}\)};
\node[fill=white] at (3, 4) {\(\Reg_{i^*, j + 1}\)};
\node[fill=white] at (2, 3) {\(\Reg_{i, j^*}\)};
\node[fill=white] at (4, 3) {\(\Reg_{i + 1, j^*}\)};
\node at (3, 3) {\(\Reg_{i^*, j^*}\)};
\end{tikzpicture}
\subcaption{The sets on a cell of a regular free-space diagram.}
\label{fig:cells}
\end{minipage}%
\begin{minipage}[t]{.5\linewidth}
\centering
\begin{tikzcd}[column sep=small,ampersand replacement=\&,baseline=(current bounding box.south)]
    \&                             \&\ar[d]              \Reg_{i^*, j + 1} \&\ar[dl] \Reg_{i + 1, j + 1}\\
\!\!\&\ar[l,dotted]  \Reg_{i, j^*} \&\ar[l]\ar[d]\ar[dl] \Reg_{i^*, j^*}   \&\ar[l]  \Reg_{i + 1, j^*}\\
    \&\ar[dl,dotted] \Reg_{i, j}   \&\ar[d,dotted]       \Reg_{i^*, j}     \&\\
~   \&~\&~\&
\end{tikzcd}
\subcaption{Dependencies of the dynamic program. \(a \to b\) means that \(a\) depends on \(b\).}
\label{fig:deps_reg}
\end{minipage}
\caption{Illustration for the dynamic program of \cref{lem:dp}.}
\label{fig:cells_deps}
\end{figure}
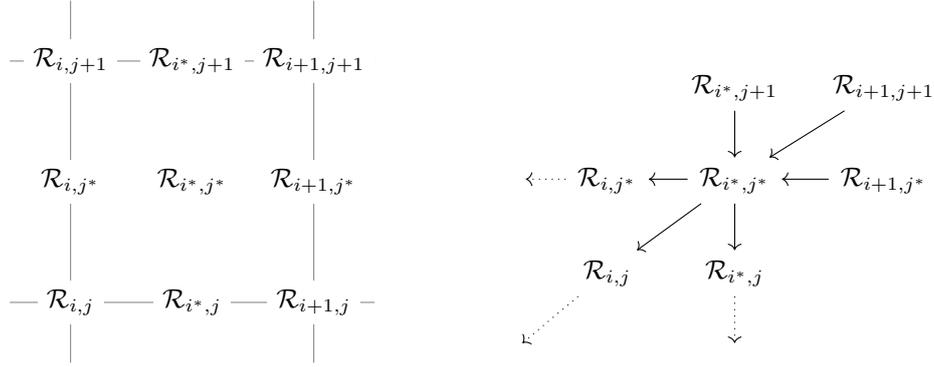

To solve the decision problem, we must decide whether \(\Reg_{m, n}\) is non-empty.
If so, then there are realisations of \(\U_m \equiv \U\) and \(\V_n \equiv \V\) that are close
enough in terms of the Fr\'echet distance.
We compute \(\Reg_{\cdot, \cdot}\) using dynamic programming.
We illustrate the propagation dependencies in \cref{fig:deps_reg} and make them explicit in
\cref{lem:dp}.

\begin{lemma}\label{lem:dp}
Let \(\star(A) := \{(p + \lambda s, q + \mu t, s, t) \mid
(p, q, s, t) \in A \text{ and } \lambda, \mu \geq 0\}\).
We have
\begin{alignat*}{2}
&\Reg_{\cdot, 0} = \Reg_{0, \cdot} &&= \emptyset\,,\\
&\Reg_{i + 1, j^*} &&= \{(p, q, s, t) \in u_{i + 1} \times \ia Y {v_{j + 1}} \times S \times S
    \mid (p, q, \ia\cdot{s}, t) \in \Reg_{i^*, j^*}\}\,,\\
&\Reg_{i^*, j + 1} &&= \{(p, q, s, t) \in \ia Y {u_{i + 1}} \times v_{j + 1} \times S \times S
    \mid (p, q, s, \ia\cdot{t}) \in \Reg_{i^*, j^*}\}\,,\\
&\Reg_{i + 1, j + 1} &&= \{(p, q, s, t) \in u_{i + 1} \times v_{j + 1} \times S \times S
    \mid (p, q, \ia\cdot{s}, \ia\cdot{t}) \in \Reg_{i^*, j^*}\}\,,\\
&\Reg_{0^*, 0^*} &&= \Fd \times S \times S\,,\\
&\Reg_{i^*, j^*}
&&= (\Fd \times S \times S) \cap \star(\Reg_{i, j} \cup \Reg_{i^*, j} \cup \Reg_{i, j^*})
\quad\text{for \(i > 0\) or \(j > 0\).}
\end{alignat*}
\end{lemma}
\begin{proof}
The first equation holds because the empty function has no parametrisation, so the Fr\'echet
distance of any pair of realisations is infinite.
The equation for \(\Reg_{i + 1, j^*}\) holds because for \(\pi \Subset \U_{i + 1}\),
\(\dir_{i + 1}(\pi) = S\), and the only additional constraint that a realisation of \(\U_{i + 1}\) has
over one of \(\U_i^*\) is that the final vertex lies in \(u_{i + 1}\).
Using symmetric properties on \(\V\), we obtain the equations for \(\Reg_{i^*, j + 1}\) and
\(\Reg_{i + 1, j + 1}\).
The equation for \(\Reg_{0^*, 0^*}\) concerns curves \(\pi\) and \(\sigma\) consisting of a single
vertex, so \(\dir_0(\pi) = \dir_0(\sigma) = S\), and \(\fr(\pi, \sigma) \leq \delta\) if and only if
\((\pi(1), \sigma(1)) \in \Fd\).
The equation for \(\Reg_{i^*, j^*}\) remains.
First we show that the right-hand side is contained in \(\Reg_{i^*, j^*}\).
Suppose that \(\pi\) and \(\sigma\) form a witness for
\((p, q, s, t) \in \Reg_{i, j} \cup \Reg_{i^*, j} \cup \Reg_{i, j^*}\).
We obtain realisations \(\pi^* \Subset \U_i^*\) and \(\sigma^* \Subset \V_j^*\) by extending the
last edge of \(\pi\) and \(\sigma\) in the direction it is already going (or adding a new edge in an
arbitrary direction if \(\pi \Subset \U_i\) or \(\sigma \Subset \V_j\)), to
\((p + \lambda s, q + \mu t)\).
If \((p + \lambda s, q + \mu t) \in \Fd\), then, by convexity of \(\Fd\), the extensions of the last
edges have Fr\'echet distance at most \(\delta\) (since the points at which the extension starts
have distance at most \(\delta\)), so \((p + \lambda s, q + \mu t, s, t) \in \Reg_{i^*, j^*}\).
Conversely, we show that the right-hand side contains \(\Reg_{i^*, j^*}\).
Let \(\pi^* \Subset \U_i^*\) and \(\sigma^* \Subset \V_j^*\) together with parametrisations
\(\alpha \colon [0, 1] \to [1, i + 1]\) and \(\beta \colon [0, 1] \to [1, j + 1]\) form a witness that
\((p, q, s, t) \in \Reg_{i^*, j^*}\).
Then, for any \(x \in [0, 1]\), the restrictions \(\pi_x\) of \(\pi^*\) and \(\sigma_x\) of
\(\sigma^*\) to the domains \([1, \alpha(x)]\) and \([1, \beta(x)]\) have Fr\'echet distance at most
\(\delta\).
Because \(\alpha\) and \(\beta\) are non-decreasing surjections, whenever \(i > 0\) or \(j > 0\),
there exists some \(x\) such that 
\begin{enumerate}
    \item \(\alpha(x) = i\) and \(\beta(x) = j\), in which case
    \(\pi_x \Subset \U_i\) and \(\sigma_x \Subset \V_j\), or
    \item \(\alpha(x) > i\) and \(\beta(x) = j\), in which case \(\pi_x \Subset \U_i^*\) and
    \(\sigma_x \Subset \V_j\), or
    \item \(\alpha(x) = i\) and \(\beta(x) > j\), in which case
    \(\pi_x \Subset \U_i\) and \(\sigma_x \Subset \V_j^*\). 
\end{enumerate}
Note that if \(i = 0\), only the second case applies, and if \(j = 0\), only the third case applies.
In each case, the last edge of \(\pi^*\) and \(\sigma^*\) extends the \(i\)-th and \(j\)-th edge of
\(\pi_x\) and \(\sigma_x\), respectively.
So \((\pi_x, \sigma_x)\) forms a witness that \((p, q, s, t)\) is contained in the right-hand side.
\end{proof}

\subparagraph{Simplifying the approach.}
Due to their dimension, the above sets can be impractical to work with.
However, for the majority of these sets, at least one of the factors \(S\) carries no additional
information, as formulated below.
Denote by \(\pr_c\) the projection map of the \(c\)-th component, so that
\(\pr_1 \colon (p, q, s, t) \mapsto p\), and in general
\(\pr_{c_1, \dots, c_k}(x) = (\pr_{c_1}(x), \dots, \pr_{c_k}(x))\).
The equations of \cref{lem:dp} imply the equivalences
\begin{align*}
(p, q, s, t) \in \Reg_{i, j\hphantom{^*}}   &\iff \hphantom{~t,}(p, q)\in \pr_{1, 2}(\Reg_{i, j})\,,\\
(p, q, s, t) \in \Reg_{i^*, j} &\iff (p, q, s) \in \pr_{1, 2, 3}(\Reg_{i^*, j})\,,\\
(p, q, s, t) \in \Reg_{i, j^*} &\iff (p, q, t) \in \pr_{1, 2, 4}(\Reg_{i, j^*})\,.
\end{align*}
Consequently, to find \(\Reg_{i, j}\), \(\Reg_{i^*, j}\), and \(\Reg_{i, j^*}\), it suffices to
compute the projections above.
This simplifies the prior dependencies as shown in \cref{fig:deps_proj}.
\begin{align*}
\pr_{1, 2}(\Reg_{i + 1, j + 1}) &= (u_{i + 1} \times v_{j + 1}) \cap \pr_{1, 2}(\Reg_{i^*, j^*})\\
&= (u_{i + 1} \times \ia Y {v_{j + 1}}) \cap \pr_{1, 2}(\Reg_{i^*, j + 1})\\
&= (\ia Y {u_{i + 1}} \times v_{j + 1}) \cap \pr_{1, 2}(\Reg_{i + 1, j^*})\,.
\end{align*}

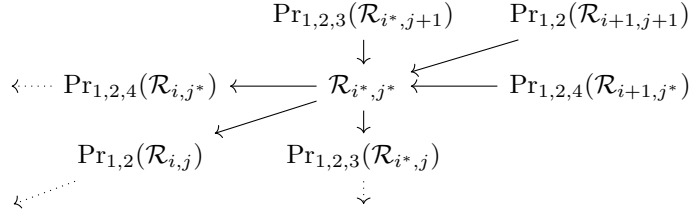
\begin{figure}
\centering
\begin{tikzcd}[column sep=small, row sep=small]
    &                                           &\ar[d] \pr_{1, 2, 3}(\Reg_{i^*, j + 1})    &\ar[dl] \pr_{1, 2}(\Reg_{i + 1, j + 1})\\
\!\!&\ar[l,dotted] \pr_{1, 2, 4}(\Reg_{i, j^*}) &\ar[l]\ar[d]\ar[dl] \Reg_{i^*, j^*}        &\ar[l] \pr_{1, 2, 4}(\Reg_{i + 1, j^*})\\
\!\!&\ar[dl,dotted] \pr_{1, 2}(\Reg_{i, j})     &\ar[d,dotted] \pr_{1, 2, 3}(\Reg_{i^*, j}) &\\
~   &~&~&
\end{tikzcd}
\caption{Simplified dependencies with projections as follows from \cref{lem:dp}.}
\label{fig:deps_proj}
\end{figure}

\subparagraph{Instantiating the approach.}
The dynamic program of \cref{lem:dp} can naturally be adapted to constrained realisations whose edge
directions are to be drawn from a subset \(S' \subseteq S^{d - 1}\), by replacing \(S\) by \(S'\),
so the framework can be used for restricted settings in~2D.
For \(S = S^{d - 1}\) the complexity of \(R_{i, j}\) can be exponential, so it can be useful to
restrict the problem.

We can look at the construction used to prove NP-hardness of the problem in~2D~\cite{icalp} as an
example for our approach.
There the curve \(\V\) is precise, so each \(v_j\) is a single point and each \(t\) is predetermined,
and curve \(\U\) consists of uncertainty regions that are vertical line segments, so each \(u_i\)
has a fixed \(x\)-coordinate and a range of \(y\)-coordinates.
If we now exclude the fixed values from our propagation, we get to track pairs \((y, s)\) of the
feasible \(y\)-coordinates on the current interval and the directions.
We start with a single region.
The hardness construction uses gadgets on the precise curve to force the uncertain curve to go
through certain points.
In our approach, this means that we keep restricting the set of feasible directions while passing
by vertices on \(\V\), and eventually each point in the starting region gives rise to two disjoint
reachable points on one of the following uncertainty regions.
So we can use our algorithm to correctly track the feasible \(y\)-coordinates through the
construction; however, we would need to keep track of regions of exponential complexity, which is,
predictably, inefficient.
Therefore, it is important to analyse the complexity of the propagated regions to determine
whether our approach gives rise to an efficient algorithm.
To illustrate our approach, we use it in the~1D case to devise an efficient algorithm in
\cref{sec:lb1d}.

\section{Lower Bound Fr\'echet Distance: One Dimension}\label{sec:lb1d}
In this section we instantiate the approach of \cref{sec:lbdd} in~1D and analyse its efficiency.
We first show the formal definitions that result from this process, and then give some intuition for
how the resulting algorithm works in~1D.

In our case, \(S = S^0\), so there are only two directions: positive \(x\)-direction and negative
\(x\)-direction.
We make use of the projections interpretation and split the projections into two regions based on
the value of the relevant direction; then all the regions we maintain are in \(\R^2\) and have a
geometric interpretation as feasible combinations of realisations of the last uncertain points on
the prefixes of the curves.
We omit \(\Reg_{i, j}\) from our computations except for checking whether \(\Reg_{m, n}\) is
non-empty.
As follows from the definition of the sets, \(\Reg_{i, j} \subseteq \Reg_{i^*, j}\) and
\(\Reg_{i, j} \subseteq \Reg_{i, j^*}\), so we can simplify the computation of
\(\Reg_{i^*, j^*}\), and then we do not need the explicit computation of \(\Reg_{i, j}\).
Furthermore, we do not compute any of \(\Reg_{i^*, j^*}\) explicitly, opting instead to substitute
them into the relevant expressions.
Therefore, we maintain the sets \(\Reg_{i, j^*}\) and \(\Reg_{i^*, j}\), splitting each into two
based on the relevant direction.
Based on our earlier free space cell interpretation (see \cref{fig:cells}), call the directions
along \(\U\) \emph{right} and \emph{left} and call the directions along \(\V\) \emph{up} and
\emph{down}.
We then have the following mapping from the regions of \cref{sec:lbdd} to the simpler intuitive
regions of this section.
\begin{align*}
U_{i, j} &= \{(p, q) \mid (p, q, \cdot, t) \in \Reg_{i, j^*} \land \ia t s = \hphantom{-}1\}\,,\\
D_{i, j} &= \{(p, q) \mid (p, q, \cdot, t) \in \Reg_{i, j^*} \land \ia t s = -1\}\,,\\
R_{i, j} &= \{(p, q) \mid (p, q, s, \cdot) \in \Reg_{i^*, j} \land s = \hphantom{-}1\}\,,\\
L_{i, j} &= \{(p, q) \mid (p, q, s, \cdot) \in \Reg_{i^*, j} \land s = -1\}\,.
\end{align*}
It is also easier to express the \(\star\) operator of \cref{lem:dp} in this setting.
Depending on which of the directions we consider fixed because we already committed to a direction,
the propagation through the cell interior works by adding either a quadrant or a half-plane to every
point in the starting region; we can denote this with a Minkowski sum.
Based on these considerations, we give the following simplified definition.

\subparagraph{Formal definition.}
Denote \(\R^{\leq 0} = \{x \in \R \mid x \leq 0\}\) and
\(\R^{\geq 0} = \{x \in \R \mid x \geq 0\}\).
Consider the space \(\R \times \R\) of the coordinates of the two curves in~1D.
We are interested in what is feasible within the \emph{interval free space,} which in this space
turns out to be a band around the line \(y = x\) of width \(2\delta\) in \(L_1\)-distance called
\(\Fd\).
For notational convenience, define the following regions (see \cref{fig:def_ii_jj}):
\[\Fd = \{(x, y) \in \R^2 \mid \lvert x - y\rvert \leq \delta\}\,,\qquad
I_i = (u_i \times \R) \cap \Fd\,,\qquad J_j = (\R \times v_j) \cap \Fd\,.\]
The propagation within the diagram consists of starting anywhere within the current region and going
in restricted directions, since we need to distinguish between going in the positive and the
negative \(x\)-direction along both curves.
We introduce the corresponding notation for restricting the directions in the form of quadrants,
half-planes, and slabs:
\[Q_{LD} = \R^{\leq 0} \times \R^{\leq 0}\,,\quad Q_{LU} = \R^{\leq 0} \times \R^{\geq 0}\,,\quad
Q_{RD} = \R^{\geq 0} \times \R^{\leq 0}\,,\quad Q_{RU} = \R^{\geq 0} \times \R^{\geq 0}\,,\]
\[H_L = \R^{\leq 0} \times \R\,,\qquad H_R = \R^{\geq 0} \times \R\,,\qquad
H_D = \R \times \R^{\leq 0}\,,\qquad H_U = \R \times \R^{\geq 0}\,.\]
\[S_L = \R^{\leq 0} \times \{0\}\,,\qquad S_R = \R^{\geq 0} \times \{0\}\,,\qquad
S_D = \{0\} \times \R^{\leq 0}\,,\qquad S_U = \{0\} \times \R^{\geq 0}\,.\]
We introduce notation for propagating in these directions from a region by taking the appropriate
Minkowski sum, denoted with \(\oplus\).
For \(a, b \in \{L, R, U, D\}\) and a region \(X\),
\[X^a = X \oplus H_a\,,\qquad X^{ab} = X \oplus Q_{ab}\,,\qquad X^{a0} = X \oplus S_a\,.\]

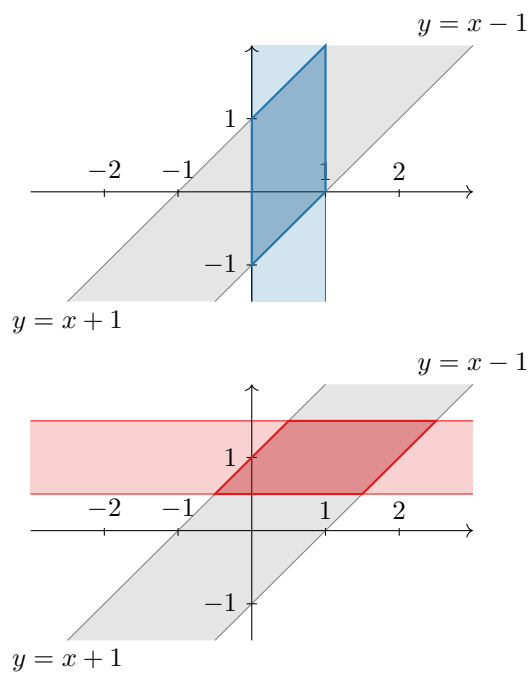
\begin{figure}
\centering
\begin{tikzpicture}[scale=.97]
\fill[myBlue,fill opacity=.2] (0, -1.5) rectangle (1, 2);
\draw[myBlue] (0, -1.5) -- (0, 2) (1, -1.5) -- (1, 2);
\fill[gray,fill opacity=.2] (-2.5, -1.5) node[below,text=black,opacity=1] {\(y = x + 1\)} --
    (1, 2) -- (3, 2) node[above,text=black,opacity=1] {\(y = x - 1\)} -- (-0.5, -1.5) -- cycle;
\draw[gray] (-2.5, -1.5) -- (1, 2) (-.5, -1.5) -- (3, 2);
\draw[thin,->] (-3, 0) -- (3, 0);
\draw[thin,->] (0, -1.5) -- (0, 2);
\foreach \x in {-2,-1,1,2} \draw (\x, -2pt) -- (\x, 1pt) node[above] {\(\x\)};
\foreach \y in {-1,1} \draw (-2pt, \y) node[left] {\(\y\)} -- (2pt, \y);
\filldraw[thick,myBlue,fill opacity=.3] (0, -1) -- (1, 0) -- (1, 2) -- (0, 1) -- cycle;
\end{tikzpicture}\hfil
\begin{tikzpicture}[scale=.97]
\fill[myRed,fill opacity=.2] (-3, .5) rectangle (3, 1.5);
\draw[myRed] (-3, .5) -- (3, .5) (-3, 1.5) -- (3, 1.5);
\fill[gray,fill opacity=.2] (-2.5, -1.5) node[below,text=black,opacity=1] {\(y = x + 1\)} --
    (1, 2) -- (3, 2) node[above,text=black,opacity=1] {\(y = x - 1\)} -- (-0.5, -1.5) -- cycle;
\draw[gray] (-2.5, -1.5) -- (1, 2) (-.5, -1.5) -- (3, 2);
\draw[thin,->] (-3, 0) -- (3, 0);
\draw[thin,->] (0, -1.5) -- (0, 2);
\foreach \x in {-2,-1,1,2} \draw (\x, -2pt) -- (\x, 1pt) node[above] {\(\x\)};
\foreach \y in {-1,1} \draw (-2pt, \y) node[left] {\(\y\)} -- (2pt, \y);
\filldraw[thick,myRed,fill opacity=.3] (-.5, .5) -- (.5, 1.5) -- (2.5, 1.5) -- (1.5, .5) -- cycle;
\end{tikzpicture}
\caption{On the left, the filled region is \(I_i = (u_i \times \R) \cap \Fd\) for \(u_i = [0, 1]\).
On the right, the filled region is \(J_j = (\R \times v_j) \cap \Fd\) for \(v_j = [0.5, 1.5]\).
In both cases \(\delta = 1\).}
\label{fig:def_ii_jj}
\end{figure}

Now we can discuss the propagation.
We start with the base case, where we compute the feasible combinations for the boundaries of the
cells of a regular free-space diagram corresponding to the first vertex on one of the curves.
For the sake of better intuition we do not use \((0, 0)\) as the base case here.
So, we fix our position to the first vertex on \(\U\) and see how far we can go along \(\V\); and
the other way around.
As we are bound to the same vertex on \(\U\), as we go along \(\V\), we keep restricting the
feasible realisations of \(u_1\).
Thus, we cut off unreachable parts of the interval as we propagate along the other curve.
We do not care about the direction we were going in after we cross a vertex on the curve where we
move.
So, if we stay at \(u_1\) and we cross over \(v_j\), then we are free to go both in the negative
and the positive direction of the \(x\)-axis to reach a realisation of \(v_{j + 1}\).
We get the following expressions, where \(U_{i, j}\) denotes the propagation upwards
from the pair of vertices \(u_i\) and \(v_j\) and propagation down, left, and right is
defined similarly:
\[U_{1, 1} = (I_1 \cap J_1)^{U0} \cap \Fd\,,\qquad D_{1, 1} = (I_1 \cap J_1)^{D0} \cap \Fd\,,\]
\[R_{1, 1} = (I_1 \cap J_1)^{R0} \cap \Fd\,,\qquad L_{1, 1} = (I_1 \cap J_1)^{L0} \cap \Fd\,,\]
\[U_{1, j + 1} = ((U_{1, j} \cup D_{1, j}) \cap J_{j + 1})^{U0} \cap \Fd\,,\qquad
D_{1, j + 1} = ((U_{1, j} \cup D_{1, j}) \cap J_{j + 1})^{D0} \cap \Fd\,,\]
\[R_{i + 1, 1} = ((R_{i, 1} \cup L_{i, 1}) \cap I_{i + 1})^{R0} \cap \Fd\,,\qquad
L_{i + 1, 1} = ((R_{i, 1} \cup L_{i, 1}) \cap I_{i + 1})^{L0} \cap \Fd\,.\]

Once the boundary regions are computed, we can proceed with propagation:
\[U_{i + 1, j} = (U_{i, j}^U \cup R_{i, j}^{RU} \cup L_{i, j}^{LU}) \cap I_{i + 1}\,,\qquad
D_{i + 1, j} = (D_{i, j}^D \cup R_{i, j}^{RD} \cup L_{i, j}^{LD}) \cap I_{i + 1}\,,\]
\[R_{i, j + 1} = (R_{i, j}^R \cup U_{i, j}^{RU} \cup D_{i, j}^{RD}) \cap J_{j + 1}\,,\qquad
L_{i, j + 1} = (L_{i, j}^L \cup U_{i, j}^{LU} \cup D_{i, j}^{LD}) \cap J_{j + 1}\,.\]

To solve the decision problem, we check if the last vertex combination is feasible:
\[((R_{m - 1, n} \cup L_{m - 1, n}) \cap I_m) \cup ((U_{m, n - 1} \cup D_{m, n - 1}) \cap J_n) \neq
\emptyset\,.\]

\subparagraph{Intuition.}
If the consecutive regions are always disjoint, we do not need to consider the possible
directions: we always know (in~1D) where the next region is, and thus what direction we take.
However, if the regions may overlap, it may be that for different realisations of a curve a segment
goes in the positive or in the negative direction.
The propagation we compute is based on the parameter space where we look at whether we have reached
a certain vertex on each curve yet, inspired by the traditional free-space diagram.
It may be that we pass by several vertices on, say, \(\V\) while moving along a single segment on
\(\U\).
The direction we choose on \(\U\) needs to be kept consistent as we compute the next regions,
otherwise we might include realisations that are invalid as feasible solutions.
Therefore, we need to keep track of the chosen direction, reflected by the pair \((s, t)\) in the
general approach and the separate sets in this section.
Otherwise, these regions in~1D are simply the feasible pairs of realisations of the last vertices on
the prefixes of the curves.

It may be helpful to think of the approach in terms of diagrams.
Consider a combination of specific vertices on the two curves, say, \(u_i\) and \(v_j\), and suppose
that we want to stay at \(u_i\) but move to \(v_{j + 1}\) on the other curve.
Which realisations of \(u_i\), \(v_j\), and \(v_{j + 1}\) can we pick that allow this move to stay
within the \(2\delta\)-band?

Suppose the \(x\)-coordinate of the diagram corresponds to the \(x\)-coordinate of \(\U\).
Then we may pick a realisation for \(u_i\) anywhere in the vertical slab corresponding to the
uncertainty interval for \(u_i\), namely, in the slab \(u_i \times \R\).
The fixed realisation for \(u_i\) would then yield a vertical line.
Now suppose the \(y\)-coordinate of the diagram corresponds to the \(x\)-coordinate of \(\V\).
For \(v_j\), picking a realisation corresponds to picking a horizontal line from the slab
\(\R \times v_j\); for \(v_{j + 1}\), it corresponds to picking a horizontal line from
\(\R \times v_{j + 1}\).
Picking a realisation for the pair \((u_i, v_j)\) thus corresponds to a point in \(u_i \times v_j\).

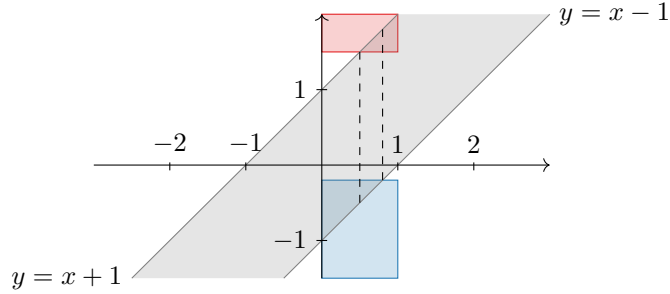
\begin{figure}
\centering
\begin{tikzpicture}
\filldraw[myBlue,fill opacity=.2] (0, -1.5) rectangle (1, -0.2);
\filldraw[myRed,fill opacity=.2] (0, 1.5) rectangle (1, 2);
\fill[gray,fill opacity=.2] (-2.5, -1.5) node[left,text=black,opacity=1] {\(y = x + 1\)} --
    (1, 2) -- (3, 2) node[right,text=black,opacity=1] {\(y = x - 1\)} -- (-0.5, -1.5) -- cycle;
\draw[gray] (-2.5, -1.5) -- (1, 2) (-.5, -1.5) -- (3, 2);
\draw[dashed] (0.8, -0.2) -- (0.8, 1.8) (0.5, -0.5) -- (0.5, 1.5);
\draw[thin,->] (-3, 0) -- (3, 0);
\draw[thin,->] (0, -1.5) -- (0, 2);
\foreach \x in {-2,-1,1,2} \draw (\x, -2pt) -- (\x, 1pt) node[above] {\(\x\)};
\foreach \y in {-1,1} \draw (-2pt, \y) node[left] {\(\y\)} -- (2pt, \y);
\end{tikzpicture}
\caption{A diagram for \(u_i = [0, 1]\), \(\textcolor{myBlue}{v_j = [-1.5, -0.2]}\), and
\(\textcolor{myRed}{v_{j + 1} = [1.5, 2]}\) with \(\delta = 1\).
Note that the feasible realisations for \(u_i\) are \([0.5, 0.8]\).}
\label{fig:simpl_prop}
\end{figure}

Of course, we may only maintain the coupling as long the distance between the coupled points is at
most \(\delta\).
For a fixed point on \(\U\), this corresponds to a \(2\delta\) window for the coordinates along
\(\V\).
Therefore, the allowed couplings are contained within the band defined by \(y = x \pm \delta\), and
when we pick the realisations for \((u_i, v_j)\), we may only pick points from \(u_i \times v_j\)
for which \(\lvert y - x\rvert \leq \delta\) holds.

As we consider the propagation to \(v_{j + 1}\), note that we may not move within \(u_i\), so the
allowed realisations for the pair \((u_i, v_{j + 1})\) are limited.
In particular, we can find that region by taking the subset of \(u_i \times v_{j + 1}\) for which
\(\lvert y - x\rvert \leq \delta\) holds and restricting the \(x\)-coordinate further to be
feasible for the pair \((u_i, v_j)\).
See \cref{fig:simpl_prop} for an illustration of this.
In this figure, we know that \(v_{j + 1}\) lies above \(v_j\); if we did not know that, we would
have to attempt propagation both upwards and downwards.
For the second curve, the same holds.

\subparagraph{Complexity.}
We now discuss the complexity of the regions we are propagating to analyse the efficiency of the
algorithm presented above.
We will perform the following steps:
\begin{enumerate}
\item Define complexity of the regions and establish the complexity of the base case.
\item Study the possible complex regions that can arise from all simple regions.
\item Study what happens to the complex regions as we propagate and conclude that the complexity is
bounded by a constant.
\end{enumerate}
The boundaries of the regions are always horizontal, vertical, or coincide with the boundaries
of \(\Fd\).
A region can be thus represented as a union of (possibly unbounded) axis-aligned rectangular
regions, further intersected with the interval free space.
We define the \emph{complexity} of a region as the minimal required number of such rectangular
regions.
Define a \emph{simple} region as a region of complexity at most~\(1\).
Observe that a simple region is necessarily convex; and a non-simple region has to be non-convex.
The illustration in \cref{fig:simple} shows the most general example of a simple region.
An empty region is also a simple region.
To enumerate the possible non-simple regions, we need to examine where higher region complexity may
come from in our algorithm.
To that aim, we first prove some simple statements about the propagation procedure.

\begin{figure}
\centering
\begin{tikzpicture}
\filldraw[fill=lightgray] (0, 0) -- (0, 2) -- (2, 4) -- (4, 4) -- (4, 2) -- (2, 0) -- cycle;
\draw[dashed] (0, 0) -- (0, -1) (0, 2) -- (0, 5) (4, -1) -- (4, 2) (4, 4) -- (4, 5);
\draw[dotted] (-1, 0) -- (0, 0) (2, 0) -- (5, 0) (-1, 4) -- (2, 4) (4, 4) -- (5, 4);
\draw (-1, 1) -- (0, 2) (2, 4) -- (3, 5) (1, -1) -- (2, 0) (4, 2) -- (5, 3);
\end{tikzpicture}
\caption{An example simple region.
We get less general ones by setting any side length to~\(0\).}
\label{fig:simple}
\end{figure}
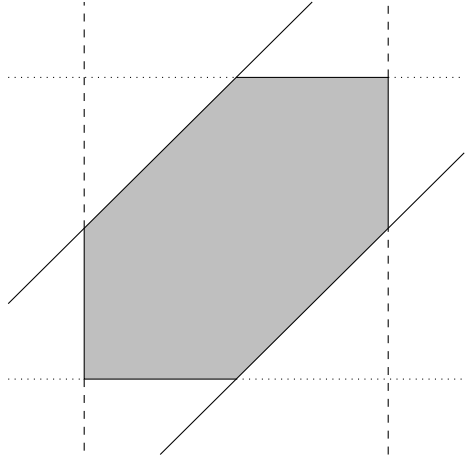

First, we discuss the complexity of the regions we can get in the base case of the propagation.
\begin{lemma}\label{lemma:base}
For all \(i \in [m - 1]\) and \(j \in [n - 1]\), regions \(U_{1, j}\), \(D_{1, j}\), \(R_{i, 1}\),
and \(L_{i, 1}\) are simple.
\end{lemma}
\begin{proof}
Consider first the intersection \(I_1 \cap J_1\).
It is the intersection of a vertical slab, a horizontal slab, and the diagonal slab (interval free
space).
All three are convex sets, hence their intersection is also convex and uses only vertical,
horizontal, and diagonal line segments, so the result is a simple region.
To obtain \(U_{1, 1}\), \(D_{1, 1}\), \(R_{1, 1}\), and \(L_{1, 1}\), we take the Minkowski sum of
the region with the corresponding half-slab.
Both are convex, so the result again is convex; we then intersect it with the interval free space
again, getting a simple region.

Now assume that \(U_{1, j}\) is simple; we show that \(U_{1, j + 1}\) is simple.
Note that for some region \(X\), \(U_{1, j} = X^{U0} \cap \Fd\) and \(D_{1, j} = X^{D0} \cap \Fd\).
Then
\[U_{1, j} \cup D_{1, j} = (X^{U0} \cup X^{D0}) \cap \Fd
= (X \oplus (\{0\} \times \R)) \cap \Fd\,.\]
So, we get a vertical slab the width of \(X\), intersected with \(\Fd\), so the result is convex.
We then intersect the region with the simple region \(J_{j + 1}\); take Minkowski sum with a slab;
and again intersect with the interval free space.
Clearly, the result is convex and uses only the allowed boundaries, so we get a simple region.

The argument for \(D_{1, j}\) is symmetric; the arguments for \(R_{i, 1}\) and \(L_{i, 1}\) are
equally straightforward.
Hence, all the regions we get in the base case are simple.
\end{proof}

To proceed, we need to make the relation in pairs \((U, D)\) and \((R, L)\) clear, so we know
where the complexity may come from.
Denote a half-plane with a vertical or horizontal boundary starting at coordinate \(s\) and going in
direction \(X\) by \(H_s^X\).
For example, a half-plane bounded on the left by the line \(x = 2\) is denoted \(H_2^R\).
\begin{lemma}\label{lemma:intersect}
Take two imprecise curves \(\U\) and \(\V\) of lengths \(m\) and \(n\), respectively, and let
\(i \in [m - 1]\) and \(j \in [n - 1]\).
Consider the pair \(R_{i, j}\), \(L_{i, j}\) and assume the regions are simple.
Then exactly one of the following options holds:
\begin{enumerate}
\item \(R_{i, j} = L_{i, j} = \emptyset\), so both regions are empty;
\item \(R_{i, j} = J_j \cap H_s^R \neq \emptyset \land L_{i, j} = \emptyset\) for some \(s\), so one
region is empty and the other spawns the entire feasible range, except that it may be cut with a
vertical line on the left;
\item \(L_{i, j} = J_j \cap H_s^L \neq \emptyset \land R_{i, j} = \emptyset\) for some \(s\), so one
region is empty and the other spawns the entire feasible range, except that it may be cut with a
vertical line on the right;
\item \(L_{i, j} \cap R_{i, j} \neq \emptyset\), so both regions are non-empty, and they intersect.
\end{enumerate}
We can make the same statement for the pair \(U_{i, j}\), \(D_{i, j}\), replacing the half-planes
with \(H_s^U\) and \(H_s^D\).
\end{lemma}
\begin{proof}
We show the statement for the pair \(R_{i, j}\), \(L_{i, j}\).
We prove the statement by induction on \(j\).
First of all, for \(j = 1\) we know that either both regions are empty (case~1), or they are both
non-empty and intersect (case~4), showing the claim.
So let \(j = j' + 1\) for the rest of the proof and assume that the lemma holds for the pair
\(R_{i, j'}\), \(L_{i, j'}\).

\begin{figure}
\centering
\begin{tikzpicture}
\fill[myRed,opacity=.2] (-1, .5) -- (1.5, .5) -- (2, 1) -- (2, 5) -- (-1, 5) -- cycle;
\fill[myBlue,opacity=.2] (1, .5) rectangle (6, 5);
\draw (-1, 0) -- (4, 5) (1, 0) -- (6, 5);
\draw (-1, 4) -- (6, 4) (-1, 3) -- (6, 3);
\draw[myRed,very thick,opacity=.4] (-.5, .5) -- (1.5, .5) -- (2, 1) -- (2, 3) -- cycle;
\draw[myBlue,very thick,opacity=.4] (4, 5) -- (1, 2) -- (1, .5) -- (1.5, .5) -- (6, 5);
\fill[myBlue,opacity=.4] (2, 3) -- (4, 3) -- (5, 4) -- (3, 4) -- cycle;
\node at (.1, .1) {\(\Fd\)};
\node at (-.7, 3.5) {\(J_j\)};
\node[myRed] at (.5, .8) {\(U_{i, j'}^{LU}\)};
\node[myBlue] at (2.5, 2.2) {\(U_{i, j'}^{RU}\)};
\node[myPurple] at (1.5, 1.5) {\(U_{i, j'}^{U0}\)};
\end{tikzpicture}
\caption{Propagation of \(U_{i, j'}\) to \(L_{i, j}\) and \(R_{i, j}\).
Note \(J_j \subset \Fd\).
Observe that \(U_{i, j'}^{LU} \cap U_{i, j'}^{RU} = U_{i, j'}^{U0}\), and if \(U_{i, j'}^{U0}
\cap J_j = \emptyset\) but \(J_j\) does not lie below \(U_{i, j'}\), then \(U_{i, j'}^{LU} \cap J_j
= \emptyset\) and \(U_{i, j'}^{RU} \cap J_j = J_j\), so one of the regions is empty and the other
covers the entire feasible region.}
\label{fig:upper_intersect}
\end{figure}
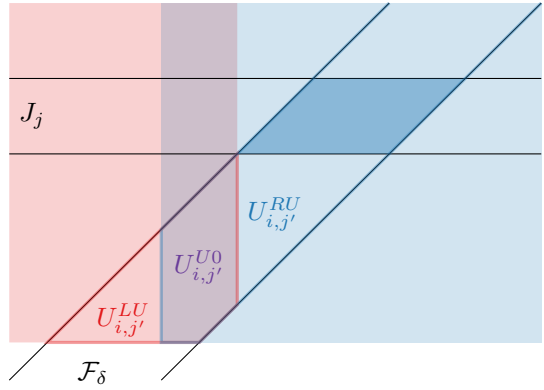

We go over the possible combinations of the previous regions that are combined in the propagation
and show that for any such combination we end up in one of the cases.
Recall that \(R_{i, j} = R_{i, j' + 1} =
(R_{i, j'}^R \cup U_{i, j'}^{RU} \cup D_{i, j'}^{RD}) \cap J_{j' + 1}\).
Similarly, \(L_{i, j} = (L_{i, j'}^L \cup U_{i, j'}^{LU} \cup D_{i, j'}^{LD}) \cap J_{j' + 1}\).
Consider the following cases:
\begin{itemize}

\item \(U_{i, j'} \neq \emptyset\).
Note that \(U_{i, j'}^{LU} \cap U_{i, j'}^{RU} = U_{i, j'}^{U0}\), so a vertical half-slab from a
lower boundary.
If \(U_{i, j'}^{U0} \cap J_{j' + 1} \neq \emptyset\), then both \(L_{i, j}\) and \(R_{i, j}\) are
non-empty and intersect, landing in case~4.
Otherwise, suppose \(U_{i, j'}^{U0} \cap J_{j' + 1} = \emptyset\).
This intersection can be empty due to two reasons.
Firstly, \(U_{i, j'}^{U0}\) may lie entirely above \(J_{j' + 1}\).
Then \(U_{i, j'}^{LU} \cap J_{j' + 1} = U_{i, j'}^{RU} \cap J_{j' + 1} = \emptyset\), so
\(U_{i, j'}\) does not contribute anything to either region; this case is considered below.
Secondly, \(U_{i, j'}^{U0}\) may lie entirely to the left of \(J_{j' + 1}\).
Then we get the situation shown in \cref{fig:upper_intersect}: it must be that
\(U_{i, j}^{LU} \cap J_{j' + 1} = \emptyset\) and \(U_{i, j}^{RU} \cap J_{j' + 1} = J_{j' + 1}\).
This means, in particular, that \(R_{i, j} = J_{j' + 1} = J_j\).
It might be that \(L_{i, j}\) and \(R_{i, j}\) are both non-empty; as \(L_{i, j} \subseteq J_j\),
they intersect, and so we end up in case~4.
Otherwise, \(L_{i, j}\) must be empty, ending up in case~2.
So, whenever \(U_{i, j'}\) contributes, we end up in one of the cases.

\item \(D_{i, j'} \neq \emptyset\).
We can make arguments symmetric to the previous case, landing us in either case~4 or case~3.
If \(D_{i, j'}\) does not contribute to either region, we consider the next case.

\item Neither \(U_{i, j'}\) nor \(D_{i, j'}\) contribute to \(L_{i, j}\) or \(R_{i, j}\), meaning 
we can simplify the expressions to \(R_{i, j' + 1} = R_{i, j'}^R \cap J_{j' + 1}\) and
\(L_{i, j' + 1} = L_{i, j'}^L \cap J_{j' + 1}\).
We use the induction hypothesis and distinguish between the cases for the pair \(R_{i, j'}\),
\(L_{i, j'}\).
Starting in case~1, we get that \(L_{i, j} = R_{i, j} = \emptyset\), ending up in case~1.
Starting in case~2, we get \(L_{i, j} = \emptyset\), and \(R_{i, j} = R_{i, j' + 1} =
J_{j' + 1} \cap R_{i, j'}^R\).
Observe that \(R_{i, j'}^R\) is a half-plane that can be denoted by \(H_s^R\) for some appropriate
\(s\); depending on whether the intersection is empty, we end in either in case~1 or in case~2.
Starting in case~3 is symmetric and lands us in either case~1 or case~3.
If we start in case~4, then the half-planes \(R_{i, j'}^R\) and \(L_{i, j'}^L\) intersect, and so
for the pair \(R_{i, j}\), \(L_{i, j}\), we end up in case~4; or in case~2 or~3 if
\(L_{i, j'}^R\cap J_{j' + 1}\) or \(R_{i, j'}^R\cap J_{j' + 1}\) is empty.
\end{itemize}
This covers all the cases.
By induction, we conclude that the lemma holds.
The proof for \(U\), \(D\) is symmetric.
\end{proof}

Let us now introduce the higher complexity regions.
\begin{definition}
A \emph{staircase with \(k\) steps} is an otherwise simple region with \(k\) cut-outs on the same
side of the region, each consisting of a single horizontal and a single vertical segments,
introducing higher complexity.
All the options for a staircase with one step (regions of complexity~2) are illustrated in
\cref{fig:staircase}.
\end{definition}
We should note that a staircase with \(k\) steps, when intersected with \(\Fd\), can yield up to
\(k + 1\) disjoint simple regions.
More specifically, every step that extends outside \(\Fd\) splits a staircase of \(k\) steps into
two staircases of at most \(k-1\) steps in total.

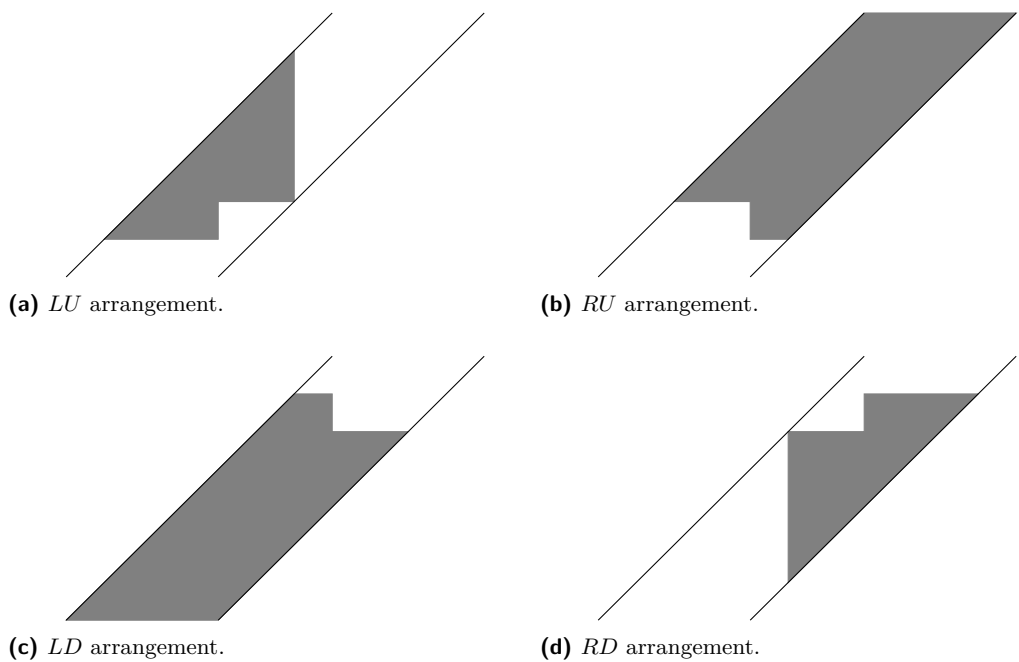
\begin{figure}
\begin{minipage}[b]{.5\linewidth}
\centering
\begin{tikzpicture}
\filldraw[gray] (0, 1) -- (1.5, 1) -- (1.5, 1.5) -- (2.5, 1.5) -- (2.5, 3.5)-- cycle;
\draw (-.5, .5) -- (3, 4) (1.5, .5) -- (5, 4);
\end{tikzpicture}
\subcaption{\(LU\) arrangement.}
\label{fig:staircase_1}
\end{minipage}%
\begin{minipage}[b]{.5\linewidth}
\centering
\begin{tikzpicture}
\filldraw[gray] (.5, 1.5) -- (1.5, 1.5) -- (1.5, 1) -- (2, 1) -- (5, 4) -- (3, 4) -- cycle;
\draw (-.5, .5) -- (3, 4) (1.5, .5) -- (5, 4);
\end{tikzpicture}
\subcaption{\(RU\) arrangement.}
\label{fig:staircase_2}
\end{minipage}%
\par\vspace{\baselineskip}%
\begin{minipage}[b]{.5\linewidth}
\centering
\begin{tikzpicture}
\filldraw[gray] (0.5, 1.5) -- (1, 1.5) -- (1, 1) -- (2, 1) -- (-.5, -1.5) -- (-2.5, -1.5) -- cycle;
\draw (-2.5, -1.5) -- (1, 2) (-.5, -1.5) -- (3, 2);
\end{tikzpicture}
\subcaption{\(LD\) arrangement.}
\label{fig:staircase_3}
\end{minipage}%
\begin{minipage}[b]{.5\linewidth}
\centering
\begin{tikzpicture}
\filldraw[gray] (0, 1) -- (1, 1) -- (1, 1.5) -- (2.5, 1.5) -- (0, -1) -- cycle;
\draw (-2.5, -1.5) -- (1, 2) (-.5, -1.5) -- (3, 2);
\end{tikzpicture}
\subcaption{\(RD\) arrangement.}
\label{fig:staircase_4}
\end{minipage}
\caption{All possible combinations for a single-step staircase.
Each can be further intersected by a vertical or horizontal slab (\(I_i\) or \(J_j\)) or shifted so
that the boundary is affected by \(\Fd\).}
\label{fig:staircase}
\end{figure}

We make the following observation relating the regions in pairs \(R_{i, j}\), \(L_{i, j}\) and
\(U_{i, j}\), \(D_{i, j}\).

\begin{lemma}\label{lemma:union_step}
Take two imprecise curves \(\U\) and \(\V\) of lengths \(m\) and \(n\), respectively, and let
\(i \in [m - 1]\) and \(j \in [n - 1]\).
Consider the pair \(R_{i, j}\), \(L_{i, j}\) and assume both regions are non-empty.
If \(j = 1\), the regions have the same \(y\)-coordinate for their lower and upper boundaries.
If \(j = j' + 1\) and the regions \(U_{i, j'}\), \(D_{i, j'}\) are simple, then the union 
\(R_{i, j} \cup L_{i, j}\) is either simple or a staircase with one step.
Furthermore, both \(R_{i, j}\) and \(L_{i, j}\) are either simple or staircases with one step.
\end{lemma}
\begin{proof}
First of all, for \(j = 1\), \cref{lemma:base} implies that \(R_{i, j}\) and \(L_{i, j}\) are
simple; furthermore, the
propagation starts from the same region, so the \(y\)-range is the same and the statement holds.
For the rest of the proof assume that \(j = j' + 1\) and regions \(U_{i, j'}\) and \(D_{i, j'}\) are
simple.

Consider the region \(R_{i, j} = (R_{i, j'}^R \cup U_{i, j'}^{RU} \cup D_{i, j'}^{RD}) \cap J_j\).
In principle, the union of the two quadrants may create a staircase with a single step.
However, as the reader may verify, adding the half-plane to the union cannot add a step, since doing
so would require a horizontal ray forming the top or the bottom boundary of the union of quadrants,
which is impossible.
Symmetrical arguments can be made for \(L_{i, j}\).
So, under the given assumptions the regions are always either simple or staircases with one step.
\Cref{fig:simpl_staircase,fig:single_step_prop_staircase} show some examples.

Now consider the union of regions \(R_{i, j} \cup L_{i, j}\):
\[R_{i, j} \cup L_{i, j} = (R_{i, j'}^R \cup L_{i, j'}^L \cup U_{i, j'}^U \cup D_{i, j'}^D) \cap J_j\,.\]
The only source of higher complexity is the union operator in the propagation.
This is the union of four half-planes.
If both \(R_{i, j'}\) and \(L_{i, j'}\) are non-empty, we know from \cref{lemma:intersect} that they
intersect, so \(J_j \subseteq R_{i, j'}^R \cup L_{i, j'}^L = \R^2\).
The same holds for the pair \(U_{i, j'}\), \(D_{i, j'}\).
Now assume that at least one region from each pair is empty, say, \(L_{i, j'}\) and \(D_{i, j'}\).
If one more region is empty, then one of \(L_{i, j}\), \(R_{i, j}\) is empty, which contradicts our
assumption.
Note that the union of two half-planes with perpendicular boundaries, intersected with a horizontal
strip and the interval free space, can create a staircase with one step.
In our particular setting we get the staircase in the \(RU\) arrangement, shown in
\cref{fig:staircase_2}.
Other choices for empty regions will give one of the other arrangements of \cref{fig:staircase}.
There are no other options, so the statement about the union \(R_{i, j} \cup L_{i, j}\) is proven.
\end{proof}

\begin{figure}
\begin{minipage}[b]{.48\linewidth}
\centering
\begin{tikzpicture}[scale=.96]
\fill[myBlue,opacity=.2] (1.5, -1) rectangle (5, 4);
\fill[myRed,opacity=.2] (-1, 1) rectangle (5, 4);
\fill[myRed,opacity=.2] (1.5, 1) rectangle (2.5, 4);
\draw[myRed] (1.5, 4) -- (1.5, 1) -- (2.5, 1) -- (2.5, 4);
\draw (1, -1) -- (6, 4) (-1, -1) -- (4, 4);
\draw (.5, -1) -- (.5, 4) (3.5, -1) -- (3.5, 4);
\draw[myRed,very thick] (1, 1) -- (2.5, 2.5) -- (2.5, 1) -- cycle;
\draw[myBlue,very thick] (1.5, 1.5) -- (3.5, 3.5) -- (3.5, 1.5) -- (1.5, -.5) -- cycle;
\draw[myPurple,very thick] (1.5, 1.5) -- (2.5, 2.5);
\node[myBlue] at (4.5, .5) {\(R_{i, j}^R\)};
\node[myRed] at (4.5, 1.5) {\(U_{i, j}^{RU}\)};
\node[myRed] at (-.5, 1.5) {\(U_{i, j}^{LU}\)};
\node[myRed] at (2, 3.5) {\(U_{i, j}^{U0}\)};
\node at (2.5, -.7) {\(I_{i + 1}\)};
\node at (0, -.7) {\(\Fd\)};
\end{tikzpicture}
\subcaption{Taking the union of \(R_{i, j}^R\) and \(U_{i, j}^{RU}\) creates a simple region.
The other region is also simple, but the union of resulting regions is a staircase.}
\label{fig:simpl_staircase}
\end{minipage}\hfill%
\begin{minipage}[b]{.48\linewidth}
\centering
\begin{tikzpicture}[scale=.96]
\fill[myBlue,opacity=.2] (-1, 0) -- (2, 0) -- (2.5, .5) -- (2.5, 4) -- (-1, 4) -- cycle;
\fill[myRed,opacity=.2] (-1, 1) rectangle (5, 4);
\draw (1, -1) -- (6, 4) (-1, -1) -- (4, 4);
\draw (.5, -1) -- (.5, 4) (3.5, -1) -- (3.5, 4);
\draw[myPurple,very thick] (.5, 0) -- (2, 0) -- (2.5, .5) -- (2.5, 1) -- (3, 1) -- (3.5, 1.5) --
                              (3.5, 3.5) -- (.5, .5) -- cycle;
\node[myBlue] at (-.5, .5) {\(L_{i, j}^{LU}\)};
\node[myRed] at (4.5, 1.5) {\(U_{i, j}^U\)};
\node[myPurple] at (2.5, 1.5) {\(U_{i + 1, j}\)};
\node at (2.5, -.7) {\(I_{i + 1}\)};
\node at (0, -.7) {\(\Fd\)};
\end{tikzpicture}
\subcaption{Taking the union of \(L_{i, j}^{LU}\) and \(U_{i, j}^U\) creates a staircase.
Intersection with \(I_{i + 1}\) preserves it: see the coloured outline of the resulting region for
\(U_{i + 1, j}\).}
\label{fig:single_step_prop_staircase}
\end{minipage}
\caption{Examples of staircase arrangements.}
\label{fig:staircase_origins}
\end{figure}

Now consider the propagation when we start from not necessarily simple regions or regions that do
not match in their \(y\)-range (or \(x\)-range), as described in \cref{lemma:union_step}.
Consider the complexity contribution when propagating across a cell\dsh say, \(U_{i, j}\) to
\(U_{i + 1, j}\).
To perform the propagation, we take
\[U_{i, j}^U = U_{i, j} \oplus H_U = U_{i, j} \oplus (\R \times \R^{\geq 0})\,.\]
From the definition of the Minkowski sum, it is easy to see that for non-empty \(U_{i, j}\) this
results in an upper half-plane with respect to the lowest point in \(U_{i, j}\).
Therefore, when propagating a region across the cell, it either contributes nothing if it is empty,
or it contributes a half-plane otherwise.
Therefore, to establish if we can arrive at progressively more complex regions, we need to
consider the other boundaries as source of complexity.
This insight together with the previous results informs the following argument.

\begin{lemma}\label{lemma:complexity}
The regions that we propagate are either simple, or staircases with one step, so the regions have
constant complexity.
\end{lemma}
\begin{proof}
As shown in \cref{lemma:base}, the base regions are always simple.
Consider the regions \(R_{i, j}\), \(L_{i, j}\) for some \(i\) and \(j = j' + 1\).
The proof for \(U_{i, j}\), \(D_{i, j}\) is symmetric.
As we have just observed, the complexity of \(R_{i, j'}\) and \(L_{i, j'}\) is irrelevant for
\(R_{i, j}\), \(L_{i, j}\), as they contribute a half-plane in the worst case.
Furthermore, we have shown in \cref{lemma:union_step} that if \(U_{i, j'}\) and \(D_{i, j'}\) are
simple, then regions \(R_{i, j}\), \(L_{i, j}\) are at worst single-step staircases.

It remains to consider what happens as we propagate further from the regions obtained in
\cref{lemma:union_step}.
So suppose \(R_{i, j}\), \(L_{i, j}\) are obtained as in \cref{lemma:union_step}.
Again, their complexity is irrelevant for the complexity of \(R_{i, j + 1}\), \(L_{i, j + 1}\), so
it remains to answer the following question.
Assuming no restrictions on \(U_{i, j}\), \(D_{i, j}\), what is the possible complexity of
\(U_{i + 1, j}\), \(D_{i + 1, j}\)?
Consider the propagation for e.g.\@
\(U_{i + 1, j} = (U_{i, j}^U \cup R_{i, j}^{RU} \cup L_{i, j}^{LU}) \cap I_{i + 1}\).
As follows from \cref{lemma:union_step} and the mechanics of propagation, the region
\(L_{i, j}^{LU} \cup R_{i, j}^{RU}\) is either a simple region or a staircase with a single step,
unbounded horizontally.
Therefore, adding the half-plane of \(U_{i, j}^U\) cannot increase the complexity.
A symmetric argument holds for \(D_{i + 1, j}\).
Hence, both \(U_{i + 1, j}\) and \(D_{i + 1, j}\) are again either simple or staircases with a
single step.

Finally, consider the propagation through the next cell to the pair \(R_{i + 1, j + 1}\),
\(L_{i + 1, j + 1}\).
For the region \(R_{i + 1, j + 1}\) we need to compute \(U_{i + 1, j}^{RU} \cup D_{i + 1, j}^{RD}\).
Note that
\[U_{i + 1, j} \cup D_{i + 1, j} = (U_{i, j}^U \cup D_{i, j}^D \cup R_{i, j}^R \cup L_{i, j}^L)
\cap I_{i + 1}\,,\]
and as both \(R_{i, j}\) and \(L_{i, j}\) are non-empty and intersect, as follows from
\cref{lemma:intersect,lemma:union_step}, we conclude \(U_{i + 1, j} \cup D_{i + 1, j} = I_{i + 1}\).
Therefore, the region \(R_{i + 1, j + 1}\) is formed with a union of two half-planes with parallel
boundaries, and so the region is simple.
The same holds for \(L_{i + 1, j + 1}\).
So, within two propagation steps we may go from simple regions to staircase regions with one step
before returning to simple regions.
As there are no other possibilities for the propagation, the statement of the lemma holds.
\end{proof}

The operations we use during propagation can be done in constant time for
constant-complexity arguments.
Using \cref{lemma:complexity}, we state the main result.
\begin{theorem}
We can solve the decision problem for lower bound Fr\'echet distance on imprecise curves
of lengths \(m\) and \(n\) in~1D in time \(\Theta(mn)\).
\end{theorem}

\section{Upper Bound Fr\'echet Distance}\label{sec:ubfr}
Until this point, we have been discussing the lower bound Fr\'echet distance.
We now turn our attention to the upper bound. 
The problem is known to be NP-hard in~2D in all variants we consider~\cite{icalp}; we show here that
this remains true even in~1D.
Define the following problems for the discrete and continuous Fr\'echet distance.
\begin{problem}
\textsc{Upper Bound (Discrete) Fr\'echet:} Given two uncertain trajectories \(\U\) and \(\V\) in 1D
of lengths \(m\) and \(n\), respectively, and a threshold \(\delta > 0\), determine if
\(\frmax(\U, \V) \leq \delta\) (\(\dfrmax(\U, \V) \leq \delta\)).
\end{problem}
We show that these problems are NP-hard both for indecisive and imprecise models by giving a
reduction from \textsc{CNF-SAT}.
The construction we use is similar to that used in~2D; however, in~2D the desired alignment of
subcurves is achieved by having one of the curves be close enough to \((0, 0)\) at all times.
Here making a curve close to \(0\) will not work, so we need to add extra gadgets instead that can
`eat up' the alignment of the subcurves that we do not care about.
We start by describing the construction and then show how it leads to the NP-hardness argument.

Suppose we are given a CNF-SAT formula \(C\) on \(n\) clauses and \(m\) variables:
\[C = \bigwedge_{i \in [n]} C_i\,,\qquad C_i = \bigvee_{j \in J \subseteq [m]} x_j \lor \bigvee_{k
\in K \subseteq [m] \setminus J} \neg x_k\quad\text{for all \(i \in [n]\).}\]
We define an \emph{assignment} as a function \(a \colon \{x_1, \dots, x_m\} \to \{\True, \False\}\)
that assigns a value to each variable, \(a(x_j) = \True\) or \(a (x_j) = \False\) for any
\(j \in [m]\).
\(C[a]\) then denotes the result of substituting \(x_j \mapsto a(x_j)\) in \(C\) for all \(j \in [m]\).
We construct two curves: curve \(\U\) is an uncertain curve that represents the variables, and curve
\(\V\) is a precise curve that represents the structure of the formula.

\subparagraph{Literal level.}
Define a \emph{literal gadget} for curve \(\V\):
\[\mathrm{LG}_{i, j} = \begin{cases}
\hphantom{-}0\hphantom{.75} \concat 1.5 & \text{if \(x_j\) is a literal of \(C_i\),}\\
-1.5\hphantom{7} \concat 1.5 & \text{if \(\neg x_j\) is a literal of \(C_i\),}\\
-0.75 \concat 1.5 & \text{otherwise.}
\end{cases}\]
Consider for now the indecisive uncertainty model.
The curve \(\U\) has an indecisive point per variable, each with two options, corresponding to
\True\ and \False\ assignments.
Define a \emph{variable gadget} for curve \(\U\):
\[\mathrm{VG}_j = \{-1.5, 0\} \concat 2.5\,.\]
Here the notation \(\{-1.5, 0\}\) denotes an indecisive point with two possible locations \(-1.5\)
and \(0\).
We interpret the position \(-1.5\) as assigning \(x_j = \True\) and the position \(0\) as assigning
\(x_j = \False\).
Observe the relationship between \(\mathrm{LG}_{i, j}\) and \(\mathrm{VG}_j\) for any given
\(i \in [n]\): the distance between the first points of the gadgets is large if the given variable
assignment turns the clause true.
For instance, if a clause has the literal \(x_j\), then the choice of \(x_j = \True\) makes the
distance between the first points \(1.5 > 1\); if the literal is \(\neg x_j\) and we make the same
choice, then the distance is \(0\); and if the literal does not occur in \(C_i\), then whichever
realisation we pick, the distance is \(0.75 < 1\).

\subparagraph{Clause level.}
We now aggregate the literal gadgets into \emph{clause gadgets.}
Similarly, we aggregate the variable gadgets into the \emph{variable section:}
\[\mathrm{CG}_{i} = 3.5 \concat \Concat_{j \in [m]} \mathrm{LG}_{i, j}\,,\qquad
\mathrm{VS} = 4.5 \concat \Concat_{j \in [m]} \mathrm{VG}_j\,.\]
Suppose that we pick some realisation for all the variables with some function \(a\).
Pick a clause \(C_i\).
Suppose that \(C_i[a] = \True\).
This means there is at least one \(x_j\) assigned in a way that makes \(C_i\) turn true.
In our construction, this means that there is at least one pair of \(\mathrm{LG}_{i, j}\) and
\(\mathrm{VG}_j\) that gives a large distance between the first two points.
If we are interested in just the Fr\'echet distance between \(\mathrm{CG}_i\) and \(\mathrm{VS}\)
for some fixed \(i\), we can state the following.
\begin{lemma}\label{lemma:clause}
For some fixed \(i \in [n]\), the (discrete) Fr\'echet distance between \(\mathrm{CG}_i\),
corresponding to clause \(C_i\), and a realisation \(\pi \Subset \mathrm{VS}\), corresponding to an
assignment \(a\), is \(1\) iff \(C_i[a] = \False\), and is \(1.5\) iff \(C_i[a] = \True\), and there
are no other possible values.
\end{lemma}
\begin{proof}
First of all, note that the points~\(4.5\) and~\(3.5\) must be matched, yielding the distance of at
least~\(1\) between the curves.
Furthermore, the only point within distance~\(1.5\) of the point~\(2.5\) that occurs at the end of
every \(\mathrm{VG}_j\) is the last point of every \(\mathrm{LG}_{i, j}\), namely,~\(1.5\).
Observe that simply walking along both curves, matching point \(k\) on one curve to point \(k\) on
the other curve for every \(k\), gives us (discrete) Fr\'echet distance of at most~\(1.5\).
Thus, the optimal matching will always match the point~\(2.5\) to one of the points at~\(1.5\).
Furthermore, the optimal solution will always match the first point of \(\mathrm{LG}_{i, j}\) to the
indecisive point of \(\mathrm{VG}_j\), as the point at~\(2.5\) is always too far.
Therefore, both for Fr\'echet and discrete Fr\'echet distance the optimal matching is one-to-one,
i.e.\@ we advance along both curves on every step.
The initial synchronisation points yield the distance~\(1\), as do the second points in the literal
level gadgets; each indecisive points is matched at distance of either~\(0\),~\(0.75\), or~\(1.5\).
The latter case only occurs if the assignment of the variable makes the clause satisfied.
So, indeed, we conclude that we can only get the distance of either~\(1\) or~\(1.5\), and the latter
is only possible if some variable turns the clause to true, so if \(C_i[a] = \True\).
Otherwise, the clause is false, and the distance is~\(1\).
\end{proof}

\subparagraph{Formula level.}
We can now paste the clause gadgets together.
Once we do that, we would like to have a way to freely choose a clause to align with the variable
section: then, if there is a clause that is not satisfied, choosing that clause would yield a small
overall distance; and if all clauses are satisfied, then any one of them will give a large distance,
and so we can distinguish between whether the formula is satisfied or not.
As a starting point, it is clear that we need to prepend and append something to the variable
section that would catch the clauses that are not aligned with the variable section.
We devise the following gadget for that:
\[\mathrm{abs} = 2.5 \concat \Concat_{j \in [m]} (-0.5 \concat 0.5)\,.\]
We show that this gadget may indeed be satisfactorily aligned with any \(\mathrm{CG}_i\).
\begin{lemma}\label{lemma:abs}
The (discrete) Fr\'echet distance between \(\mathrm{abs}\) and any \(\mathrm{CG}_i\) is~\(1\).
\end{lemma}
\begin{proof}
First of all, note that we must match the first synchronisation point of \(\mathrm{CG}_i\)
at~\(3.5\) to some point on the other curve, and the only point in \(\mathrm{abs}\) that is close
enough is the point at~\(2.5\) in the beginning.
This establishes the lower bound of~\(1\).
Furthermore, we can always get the distance of~\(1\) by walking step-by-step along both curves:
the distance between any of~\(-1.5\),~\(-0.75\), and~\(0\) is at most~\(1\) to~\(-0.5\), and the
distance between~\(1.5\) and~\(0.5\) is~\(1\).
Thus, the statement holds.
\end{proof}
We need as many of these gadgets as there may be misaligned clauses.
In the worst case, we may align \(\mathrm{CG}_1\) or \(\mathrm{CG}_n\) with \(\mathrm{VS}\), and so
we need \(n - 1\) of the catch gadgets before and after \(\mathrm{VS}\).
However, the new problem we get is that now the extra \(\mathrm{abs}\) clauses need to be aligned
with something.
To that end, we devise the following gadget:
\[\mathrm{abs}^2 = 1.5 \concat 0.5\,.\]
Again, we show that it can perform its function.
\begin{lemma}\label{lemma:abs2}
The (discrete) Fr\'echet distance between \(\mathrm{abs}^2\) and \(\mathrm{abs}\) is~\(1\).
\end{lemma}
\begin{proof}
First of all, note that we must match the first synchronisation point of \(\mathrm{abs}\)
at~\(2.5\) to the point at~\(1.5\) on \(\mathrm{abs}^2\), giving the lower bound of~\(1\).
Furthermore, we can always get the distance of~\(1\) by stepping to the second point on both curves
and staying at~\(0.5\) on \(\mathrm{abs}^2\) while alternating between~\(-0.5\) and~\(0.5\) on
\(\mathrm{abs}\).
Thus, the statement holds.
\end{proof}
Finally, we need to align these gadgets with something, but that is not too difficult, as they only
have the length of \(1\).
We define our final uncertain curves:
\[\U = 1 \concat \Big(\Concat_{\mathclap{i \in [n - 1]}} \mathrm{abs}\Big) \concat
\mathrm{VS} \concat \Big(\Concat_{\mathclap{i \in [n - 1]}} \mathrm{abs}\Big) \concat 1\,,\qquad
\V = \Big(\Concat_{\mathclap{i \in [n - 1]}} \mathrm{abs}^2\Big) \concat
\Big(\Concat_{\mathclap{i \in [n]}} \mathrm{CG}_i\Big) \concat
\Big(\Concat_{\mathclap{i \in [n - 1]}} \mathrm{abs}^2\Big)\,.\]
We illustrate the curves in \cref{fig:ub_construction}.
With these definitions, we can show the following.

{\def\y{-.3}
\begin{figure}
\centering
\begin{tikzpicture}[xscale=.38,yscale=.85]
\draw (-3, \y) -- (18, \y) (-3, 8*\y) -- (18, 8*\y) (-3, 14*\y) -- (18, 14*\y)
    (-3, 21*\y) -- (18, 21*\y);
\node at (-2.5, 11*\y) {\(\mathrm{VS}\)};
\node at (-2.5, 4.5*\y) {\(\mathrm{abs}\)};
\node at (-2.5, 17.5*\y) {\(\mathrm{abs}\)};
\draw[dashed] (-1.5, 0) -- (-1.5, 24*\y) node[below] {\(-1.5\)}
    (0, 0) -- (0, 24*\y) node[below] {\(0\)} (1.5, 0) -- (1.5, 24*\y) node[below] {\(1.5\)}
    (2.5, 0) node[above] {\(2.5\)} -- (2.5, 24*\y) (3.5, 0) -- (3.5, 24*\y) node[below] {\(3.5\)}
    (4.5, 0) node[above] {\(4.5\)} -- (4.5, 24*\y);
\draw[myBlue,thick] (1, 0) -- (2.5, \y) -- (-.5, 2*\y) -- (.5, 3*\y) -- (-.5, 4*\y) -- (.5, 5*\y)
    -- (-.5, 6*\y) -- (.5, 7*\y) -- (4.5, 8*\y) -- (-1.5, 9*\y) -- (2.5, 10*\y) -- (-1.5, 11*\y)
    -- (2.5, 12*\y) -- (0, 13*\y) -- (2.5, 14*\y) -- (2.5, 15*\y) -- (-.5, 16*\y) -- (.5, 17*\y)
    -- (-.5, 18*\y) -- (.5, 19*\y) -- (-.5, 20*\y) -- (.5, 21*\y) -- (1, 22*\y);

\begin{scope}[xshift=7.5cm]
\draw[dashed] (-1.5, 0) -- (-1.5, 24*\y) node[below] {\(-1.5\)}
    (0, 0) -- (0, 24*\y) node[below] {\(0\)} (1.5, 0) -- (1.5, 24*\y) node[below] {\(1.5\)}
    (2.5, 0) node[above] {\(2.5\)} -- (2.5, 24*\y) (3.5, 0) -- (3.5, 24*\y) node[below] {\(3.5\)};
\end{scope}
\begin{scope}[xshift=7.5cm,yshift=6*\y cm]
\draw[myRed,thick] (1.5, 0) -- (.5, \y) -- (3.5, 2*\y) -- (0, 3*\y) -- (1.5, 4*\y) -- (-.75, 5*\y)
    -- (1.5, 6*\y) -- (0, 7*\y) -- (1.5, 8*\y) -- (3.5, 9*\y) -- (-1.5, 10*\y) -- (1.5, 11*\y)
    -- (0, 12*\y) -- (1.5, 13*\y) -- (-1.5, 14*\y) -- (1.5, 15*\y) -- (1.5, 16*\y) -- (.5, 17*\y);
\end{scope}

\begin{scope}[xshift=14.5cm]
\draw[dashed] (-1.5, 0) -- (-1.5, 24*\y) node[below] {\(-1.5\)}
    (0, 0) -- (0, 24*\y) node[below] {\(0\)} (1.5, 0) -- (1.5, 24*\y) node[below] {\(1.5\)}
    (2.5, 0) node[above] {\(2.5\)} -- (2.5, 24*\y) (3.5, 0) -- (3.5, 24*\y) node[below] {\(3.5\)};
\end{scope}
\begin{scope}[xshift=14.5cm,yshift=-1*\y cm]
\draw[myRed,thick] (1.5, 0) -- (.5, \y) -- (3.5, 2*\y) -- (0, 3*\y) -- (1.5, 4*\y) -- (-.75, 5*\y)
    -- (1.5, 6*\y) -- (0, 7*\y) -- (1.5, 8*\y) -- (3.5, 9*\y) -- (-1.5, 10*\y) -- (1.5, 11*\y)
    -- (0, 12*\y) -- (1.5, 13*\y) -- (-1.5, 14*\y) -- (1.5, 15*\y) -- (1.5, 16*\y) -- (.5, 17*\y);
\end{scope}
\end{tikzpicture}\hfil
\includegraphics{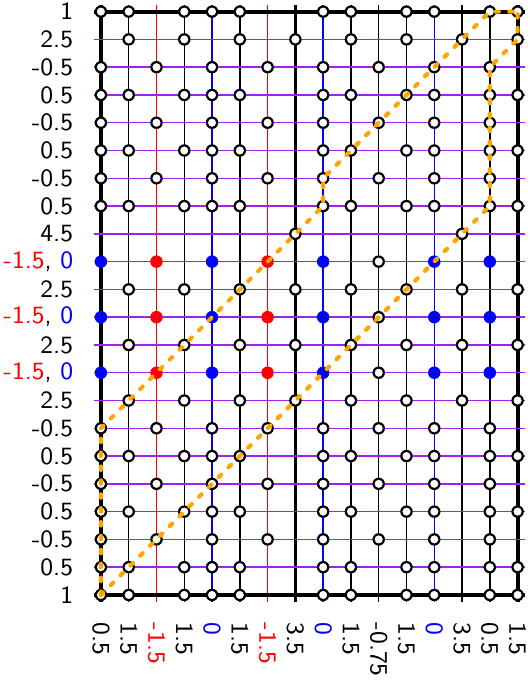}
\caption{
Left: The \textcolor{myBlue}{realisation of \(\U\)} for assignment
\(x_1 = \True\), \(x_2 = \True\), \(x_3 = \False\)
and \textcolor{myRed}{curve \(\V\)} for the formula
\(C = (x_1 \lor x_3) \land (\neg x_1 \lor x_2 \lor \neg x_3)\).
Note that \(C = \True\) with this assignment, and that both feasible alignments give (discrete)
Fr\'echet distance of~\(1.5\).
Right:
The corresponding free space.
White dots are accessible, spots without a dot are never accessible.
Blue (red) dots are only accessible if the corresponding variable is set to \True\ (\False).
Yellow dashed paths indicate potential paths through the free space; the goal is to determine
if the variables can be set such that all potential paths are blocked.
}
\label{fig:ub_construction}
\end{figure}}

\begin{theorem}
The problem \textsl{\textsc{Upper Bound (Discrete) Fr\'echet}} is NP-hard in the indecisive model.
\end{theorem}
\begin{proof}
First of all, notice that in our construction the synchronisation points at the start of the clauses
gadgets must be matched to the synchronisation points at the start of the variable section and the
\(\mathrm{abs}\) gadgets in the optimal matching, as hinted at in the proofs of the previous lemmas.
Furthermore, note that any number of \(\mathrm{abs}^2\) at the start can be matched to the point
at~\(1\), and any number of \(\mathrm{abs}^2\) can be matched to the point at~\(1\) at the end.
Putting these observations together with \cref{lemma:clause,lemma:abs,lemma:abs2}, it is easy to see
the following.
Choose some assignment \(a\) and consider the corresponding realisation \(\pi \Subset \U\).
Suppose that \(C[a] = \False\); this means that there is at least one \(i\) for which
\(C_i[a] = \False\).
Our construction allows us to consider the alignment where we match \(\mathrm{CG}_i\) to
\(\mathrm{VS}\), and the rest of clauses to one of \(\mathrm{abs}\); the remaining \(\mathrm{abs}\)
are matched to the \(\mathrm{abs}^2\), and the remaining \(\mathrm{abs}^2\) are matched to~\(1\).
In this matching, the (discrete) Fr\'echet distance between the curves is~\(1\), which is optimal,
and the formula is not satisfied.
Now suppose that \(C[a] = \True\); this means that for all \(i\), \(C_i[a] = \True\).
So, no matter which \(\mathrm{CG}_i\) we choose to align with \(\mathrm{VS}\), we get the distance
of~\(1.5\); therefore, the formula is satisfied, and the optimal distance is~\(1.5\).

Recall that the upper bound distance takes the maximum distance over all realisations.
Therefore, if the upper bound distance is~\(1\), then all the realisations yield the distance~\(1\),
and so all assignments \(a\) yield \(C[a] = \False\), and the formula is not satisfiable.
On the other hand, if the upper bound distance is~\(1.5\), then there is some realisation that
yields this distance, and it corresponds to an assignment \(a\) with \(C[a] = \True\), so the
formula is satisfiable.
Thus, our construction with the threshold \(\delta = 1\) solves \textsc{CNF-SAT}.
Curve \(\U\) has length \(2 + 2 \cdot (n - 1) \cdot (1 + 2m) + 1 + 2m = 2n + 4mn - 2m + 1\);
curve \(\V\) has length \(4 \cdot (n - 1) + n \cdot (1 + 2m) = 5n + 2mn - 4\).
Clearly, the construction takes polynomial time.
Therefore, the problem both for discrete and continuous Fr\'echet distance is NP-hard.
\end{proof}

We can easily extend this result to the imprecise curves.
We replace the indecisive points at \(\{-1.5, 0\}\) with intervals \([-1.5, 0]\).
The following observation is key.
\begin{observation}
Any upper bound solution that can be found as a certificate in the construction with the indecisive
points can also be found in the imprecise construction.
\end{observation}
Furthermore, note that no realisation can yield a distance above~\(1.5\) with an optimal matching.
Thus, if the formula is satisfiable, the upper bound distance is still~\(1.5\), and this distance
cannot be obtained otherwise.
We conclude that the problem is NP-hard.
\begin{theorem}
The problem \textsl{\textsc{Upper Bound (Discrete) Fr\'echet}} is NP-hard in the imprecise model.
\end{theorem}

\section{Weak Fr\'echet Distance}
In this section, we investigate the weak Fr\'echet distance for uncertain curves.
In general, since weak matchings can revisit parts of the curve, the dynamic program for the regular
Fr\'echet distance cannot easily be adapted, as it relies on the fact that only the realisation of
the last few vertices is tracked.
In particular, when computing the weak Fr\'echet distance for uncertain curves, one cannot simply
forget the realisations of previously visited vertices, as the matching might revisit them.
Surprisingly, we can show that for the continuous weak Fr\'echet distance between uncertain
one-dimensional curves, we can still obtain a polynomial-time dynamic program, as shown in
\cref{sec:w1dalg}.
One may expect that the discrete weak Fr\'echet distance for uncertain curves in~1D is
also solvable in polynomial time; however, in \cref{sec:wdhard} we show that this problem is
NP-hard.
We also show that computing the continuous weak Fr\'echet distance is NP-hard for uncertain
curves in~2D.

\subsection{Algorithm for Continuous Setting}\label{sec:w1dalg}
We first introduce some definitions.
Consider polygonal one-dimensional curves \(\pi \colon [1, m] \to \R\) and
\(\sigma \colon [1, n] \to \R\) with vertices at the integer parameters.
Let \(\rev{\pi}\) denote the reversal of a polygonal curve \(\pi\).
Denote by \(\pi|_{[a, b]}\) the restriction of \(\pi\) to the domain \([a, b]\).
For integer values of \(a\) and \(b\), note that \(\pi|_{[a, b]} \equiv \pi[a: b]\).
Finally, define the \emph{image} of a curve as the set of points in \(\R\) that belong to the curve,
\(\Ima(\pi) \equiv \{\pi(x) \mid x \in [1, m]\}\) for \(\pi \colon [1, m] \to \R\).
For any polygonal curve \(\pi\), define the \emph{growing curve} \(\grow{\pi}\) of \(\pi\) as the
sequence of local minima and maxima of the sequence
\(\langle\pi(i) \mid \pi(i) \notin \Ima(\pi|_{[1, i)})\rangle_{i = 1}^m\).
Thus, the vertices of a growing curve alternate between local minima and maxima, the subsequence of
local maxima is strictly increasing, and the subsequence of local minima is strictly decreasing.

It has been shown that for precise one-dimensional curves, the weak Fr\'echet distance can be
computed in linear time~\cite{buchin:2019}.
For uncertain curves, it is unclear how to use that linear-time algorithm; however,
we can apply some of the underlying ideas.
A \emph{relaxed matching} between \(\pi\) and \(\sigma\) is defined by parametrisations
\(\alpha \colon [0, 1] \to [1, m]\) and \(\beta \colon [0, 1] \to [1, n]\) with
\(\alpha(0) = 1\), \(\alpha(1) = x \in [m - 1, m]\) and
\(\beta(0) = 1\), \(\beta(1) = y \in [n - 1, n]\).
Observe that the final points of parametrisations have to be on the last segments of the curves, but
not necessarily at the endpoints of those segments.
Moreover, define a relaxed matching \((\alpha, \beta)\) to be \emph{cell-monotone} if for all
\(t \leq t'\), we have \(\min(\lfloor\alpha(t)\rfloor, m - 1) \leq \alpha(t')\) and
\(\min(\lfloor\beta(t)\rfloor, n - 1) \leq \beta(t')\).
In other words, once we pass by a vertex to the next segment on a curve, we do not allow going back
to the previous segment; backtracking within a segment is allowed.
Let \(\rMono(\pi, \sigma)\) be the minimum matching cost over all cell-monotone relaxed matchings:
\[\rMono(\pi, \sigma) = \inf_{\text{cell-monotone relaxed matching }\mu} \cost_\mu(\pi, \sigma)\,.\]
It has been shown for precise curves~\cite{buchin:2019} that
\[\wfr(\pi, \sigma) = \max\big(\rMono(\grow{\pi}, \grow{\sigma}),
\rMono(\grow{\rev{\pi}}, \grow{\rev{\sigma}})\big)\,.\]
Let \(\rMono(\pi, \sigma)[i, j] \equiv \rMono(\pi[1: i], \sigma[1: j])\).
Then the value of \(\rMono(\pi, \sigma)\) can be computed in quadratic time as
\(\rMono(\pi, \sigma)[m, n]\) using the following dynamic program:
\begin{alignat*}{2}
&\rMono(\pi, \sigma)[0, \cdot] = \rMono(\pi, \sigma)[\cdot, 0] &&= \infty\,,\\
&\rMono(\pi, \sigma)[1, 1] &&= \lvert\pi(1) - \sigma(1)\rvert\,, \text{ and for \(i > 0\) or \(j > 0\),}\\
&\rMono(\pi, \sigma)[i + 1, j + 1] &&= \min\left\{\begin{aligned}
    \max\big(\rMono(\pi, \sigma)[i, j + 1], d\big(\pi(i), \Ima(\sigma[j: j + 1])\big)\big)\,,\\
    \max\big(\rMono(\pi, \sigma)[i + 1, j], d\big(\sigma(j), \Ima(\pi[i: i + 1])\big)\big)\,.
\end{aligned}\right.
\end{alignat*}
If \(\pi\) is a growing curve, we have \(\Ima(\pi[i, i + 1]) = \Ima(\pi[1: i + 1])\), so the
following dynamic program is equivalent if \(\pi\) and \(\sigma\) are growing curves:
\begin{alignat*}{2}
&r(\pi, \sigma)[0, \cdot] = r(\pi, \sigma)[\cdot, 0] &&= \infty\,,\\
&r(\pi, \sigma)[1, 1] &&= \lvert\pi(1) - \sigma(1)\rvert\,,
\text{ and for \(i > 0\) or \(j > 0\),}\\
&r(\pi, \sigma)[i + 1, j + 1] &&= \min\begin{cases}
    \max\big(r(\pi, \sigma)[i, j + 1], d\big(\pi(i), \Ima(\sigma[1: j + 1])\big)\big)\,,\\
    \max\big(r(\pi, \sigma)[i + 1, j], d\big(\sigma(j), \Ima(\pi[1: i + 1])\big)\big)\,.
\end{cases}
\end{alignat*}
Let \(r(\pi, \sigma) \coloneqq r(\pi, \sigma)[m, n]\) when executing the dynamic program above for
curves \(\pi \colon [1, m] \to \R\) and \(\sigma \colon [1, n] \to \R\).
We have \(\rMono(\grow{\pi}, \grow{\sigma}) = r(\grow{\pi}, \grow{\sigma})\).
Moreover, observe that the final result of computing \(r\) is the same whether we apply it to the
original or the growing curves.
In other words, \(r(\pi, \sigma) = r(\grow{\pi}, \grow{\sigma})\), so
\begin{align*}
\wfr(\pi, \sigma)
&= \max\big(\rMono(\grow{\pi}, \grow{\sigma}), \rMono(\grow{\rev{\pi}}, \grow{\rev{\sigma}})\big)\\
&= \max\big(r(\grow{\pi}, \grow{\sigma}), r(\grow{\rev{\pi}}, \grow{\rev{\sigma}})\big)\\
&= \max\big(r(\pi, \sigma), r(\rev{\pi}, \rev{\sigma})\big)\,.
\end{align*}

With regard to computing the minimum weak Fr\'echet distance over realisations of uncertain curves,
this roughly means that we only need to keep track of the image of the prefix (and he suffix) of
\(\pi\) and \(\sigma\).
To formalise this, we split up the computation over the prefix and the suffix.
Let \(i_{\min}, i_{\max} \in [m]\), \(j_{\min}, j_{\max} \in [n]\),
\([x_{\min}, x_{\max}] \subseteq \R\), and \([y_{\min}, y_{\max}] \subseteq \R\).
Abbreviate the pairs \(I := (i_{\min}, i_{\max})\), \(J := (j_{\min}, j_{\max})\) and the intervals
\(X := [x_{\min}, x_{\max}]\), \(Y := [y_{\min}, y_{\max}]\), and call a realisation \(\pi\) of an
uncertain curve \emph{\(I\)-respecting} if \(\pi(i_{\min})\) is a global minimum of \(\pi\) and
\(\pi(i_{\max})\) is a global maximum of \(\pi\).
Moreover, say that \(\pi\) is \emph{\((I, X)\)-respecting} if additionally
\(\pi(i_{\min}) = x_{\min}\) and \(\pi(i_{\max}) = x_{\max}\).
Let \(\pi' \Subset \U_I\) and \(\pi'' \Subset \U_I^X\) denote some \(I\)- and \((I, X)\)-respecting
realisations of an uncertain curve \(\U\), respectively.
Consider the minimum weak Fr\'echet distance between \((I, X)\)- and \((J, Y)\)-respecting
realisations \(\pi \Subset \U_I^X\) and \(\sigma \Subset \V_J^Y\):
\[\wfrmin(\U_I^X, \V_J^Y) \equiv \min_{\pi \Subset \U_I^X, \sigma \Subset \V_J^Y} \wfr(\pi, \sigma)
= \min_{\pi \Subset \U_I^X, \sigma \Subset \V_J^Y} \max\big(r(\pi, \sigma),
r(\rev{\pi}, \rev{\sigma})\big)\,.\]

\begin{lemma}\label{lem:weakIndependent}
Among \((I, X)\)- and \((J, Y)\)-respecting realisations, the prefix and the suffix are independent:
\[\wfrmin(\U_I^X, \V_J^Y) = \max\begin{cases}
    \min_{\pi \Subset \U_I^X, \sigma \Subset \V_J^Y} r(\pi, \sigma)\,,\\
    \min_{\pi' \Subset \U_I^X, \sigma' \Subset \V_J^Y} r(\rev{\pi'}, \rev{\sigma'})\,.
\end{cases}\]
\end{lemma}
\begin{proof}
If we take \(\pi = \pi'\) and \(\sigma = \sigma'\), the right-hand side becomes a lower bound on
\(\wfrmin(\U_I^X, \V_J^Y)\).
To show that it is also an upper bound, consider \((I, X)\)-respecting realisations \(\pi\) and
\(\pi'\), and define \(\pi_c\) as the prefix of \(\pi\) up to \(i_{\min}\) concatenated with the
suffix of \(\pi'\) starting from \(i_{\min}\).
Then \(\pi_c\) is an \((I, X)\)-respecting realisation of \(\U\).
Moreover, the growing curves \(\grow{\pi}\) and \(\grow{\pi_c}\) are the same (this is obvious if
\(i_{\min} > i_{\max}\), and follows from the fact that the value of the \(i_{\max}\)-th vertex is
\(x_{\max}\) otherwise).
Symmetrically, \(\grow{\rev{\pi'}} = \grow{\rev{\pi_c}}\).
We can similarly define a \((J, Y)\)-respecting realisation \(\sigma_c\) of \(\V\) based on some
\(\sigma\) and \(\sigma'\).
Since \(\grow{\pi} = \grow{\pi_c}\) and \(\grow{\sigma} = \grow{\sigma_c}\), we have
\(r(\pi, \sigma) = r(\pi_c, \sigma_c)\), and symmetrically,
\(r(\rev{\pi'}, \rev{\sigma'}) = r(\rev{\pi_c}, \rev{\sigma_c})\).
We can therefore use \(\pi_c \Subset \U_I^X\) and \(\sigma_c \Subset \V_J^Y\) in the definition of
\(\wfr(\U_I^X, \V_J^Y)\) to obtain the desired upper bound.
\end{proof}
% Before we show how to compute \(\wfrmin(\U, \V)\) in general, we first handle the (symmetric) cases
% \(m = 1\) and \(n = 1\), that are not supported by the algorithm above.
% For \(m = 1\), computing \(\wfrmin(\U, \V)\) corresponds to finding a point \(x \in u_1\) and a
% realisation \(\sigma \Subset \V\) minimising \(d(x, \Ima(\sigma))\).
% For a given \(x\), such a realisation can be found by realising each \(v_j\) as close to \(x\) as
% possible, so \(\wfrmin(\U, \V) = \min_{x \in u_1} \max_j d(x, v_j)\).
% The term \(\max_j d(x, v_j)\) is the upper envelope of distance functions \(x \mapsto d(x, v_j)\),
% and has complexity linear in the total complexity of \(\V\), and can be computed in polynomial time.
% Finding the minimising \(x \in u_1\) is trivial given the envelope.
% The remainder of this section concerns the case where \(m, n > 1\) and is guided by the following
% observations based on \cref{lem:weakIndependent}.
The remainder of this section is guided by the following observations based on
\cref{lem:weakIndependent}.
\begin{enumerate}
\item If we can compute \(\min_{\pi \Subset \U_I^X, \sigma \Subset \V_J^Y} r(\pi, \sigma)\), we can
compute \(\wfrmin(\U_I^X, \V_J^Y)\).
\item To compute \(\wfrmin(\U_I, \V_J)\), we must find an optimal pair of images \(X\) and \(Y\)
for \(\pi\) and \(\sigma\).
\item We can find \(\wfrmin(\U, \V)\) by computing \(\wfrmin(\U_I, \V_J)\) for all \(O(m^2 n^2)\)
values for \((I, J)\).
\end{enumerate}
Instead of computing \(\min_{\pi \Subset \U_I^X, \sigma \Subset \V_J^Y} r(\pi, \sigma)\) for a
specific value of \((X, Y)\), we compute the function
\((X, Y) \mapsto \min_{\pi \Subset \U_I^X, \sigma \Subset \V_J^Y} r(\pi, \sigma)\) using a dynamic
program that effectively simulates the dynamic program \(r(\pi, \sigma)\) for all \(I\)- and
\(J\)-respecting realisations simultaneously.
So let
\begin{align*}
R_{I, J}[i, j](x, y, X, Y)
:=& \inf_{\substack{\pi \Subset \U_I, \Ima(\pi[1: i]) = X, \pi(i) = x\\
\sigma \Subset \V_J, \Ima(\sigma[1: j]) = Y, \sigma(j) = y}} r(\pi, \sigma)[i, j],\quad\text{then}\\
R_{I, J}[m, n](x, y, X, Y)
=& \inf_{\substack{\pi \Subset \U_I^X, \pi(m) = x\\\sigma \Subset \V_J^Y, \sigma(n) = y}}
r(\pi, \sigma).
\end{align*}

We derive
\begin{align*}
&R_{I, J}[0, \cdot](x, y, X, Y) = R_{I, J}[\cdot, 0](x, y, X, Y) = \infty,\\
&R_{I, J}[1, 1](x, y, X, Y) = \inf_{\substack{\pi \Subset \U_I, \{x\} = X, \pi(1) = x\\
\sigma \Subset \V_J, \{y\} = Y, \sigma(1) = y}} \lvert\pi(1) - \sigma(1)\rvert,
\text{ and for \((i, j)\neq(1, 1)\)}\\
&R_{I, J}[i, j](x, y, X, Y)\\
&\quad= \inf_{\substack{\pi \Subset \U_I, \Ima(\pi[1: i]) = X, \pi(i) = x\\
    \sigma \Subset \V_J, \Ima(\sigma[1: j]) = Y, \sigma(j) = y}}
\min\begin{cases}
    \max \{r(\pi, \sigma)[i - 1, j], d(\pi(i - 1), Y)\},\\
    \max \{r(\pi, \sigma)[i, j - 1], d(\sigma(j - 1), X)\}
\end{cases}\\
&\quad= \min\begin{cases}
\inf_{\substack{\pi \Subset \U_I, \Ima(\pi[1: i]) = X, \pi(i) = x\\
    \sigma \Subset \V_J, \Ima(\sigma[1: j]) = Y, \sigma(j) = y\\
    \pi(i - 1) = x'}}
    \max \{r(\pi, \sigma)[i - 1, j], d(x', Y)\},\\
\inf_{\substack{\pi \Subset \U_I, \Ima(\pi[1: i]) = X, \pi(i) = x\\
    \sigma \Subset \V_J, \Ima(\sigma[1: j]) = Y, \sigma(j) = y\\
    \sigma(j - 1) = y'}}
    \max \{r(\pi, \sigma)[i, j - 1], d(y', X)\}
\end{cases}\\
&\quad= \min\begin{cases}
\inf_{\substack{\pi \Subset \U_I, \Ima(\pi[1: i]) = X, \pi(i) = x\\
    \Ima(\pi[1: i - 1]) = X', \pi(i - 1) = x'}}
    \max \{R_{I, J}[i - 1, j](x', y, X', Y), d(x', Y)\},\\
\inf_{\substack{\sigma \Subset \V_J, \Ima(\sigma[1: j]) = Y, \sigma(j) = y\\
    \Ima(\sigma[1: j - 1]) = Y', \sigma(j - 1) = y'}}
    \max \{R_{I, J}[i, j - 1](x, y', X, Y'), d(y', X)\},
\end{cases}
\end{align*}
where, crucially, the conditions on \(x'\), \(y'\), \(X'\), and \(Y'\) can be checked purely in
terms of \(\U_I\) and \(\V_J\), so the recurrence does not depend on any particular \(\pi\) or
\(\sigma\).
This yields a dynamic program that constructs the function \(R_{I, J}[i, j]\) based on the functions
\(R_{I, J}[i - 1, j]\) and \(R_{I, J}[i, j - 1]\).
 
The recurrence has the parameters \(I\), \(J\), \(i\), \(j\), \(x\), \(y\), \(X\), and \(Y\).
The first four are easy to handle, since \(i \in [m]\), \(j \in [n]\), \(I \in [m]^2\), and
\(J \in [n]^2\).
The other parameters are continuous.
\(X\) can be represented by \(x_{\min}\) and \(x_{\max}\), \(Y\) by \(y_{\min}\) and \(y_{\max}\).
To prove that we can solve the recurrence in polynomial time, it is sufficient to prove that we can
restrict the computation to a polynomial number of different \(x_{\min}\), \(x_{\max}\),
\(y_{\min}\), \(y_{\max}\), \(x\) and \(y\).

We assume that each of the \(u_i\) and \(v_j\) is given as a set of intervals.
This includes the cases of uncertain curves with imprecise vertices (where each of these is just one
interval) and with indecisive vertices (where each interval is just a point; but in this case we get
by definition only a polynomial number of different values for the parameters).

Consider the realisations \(\pi = \langle p_1, \dots, p_m\rangle\) and
\(\sigma = \langle q_1, \dots, q_n\rangle\) of the curves that attain the lower bound weak Fr\'echet
distance \(\wfrmin(\U, \V) =: \delta\).
In these realisations, we need to have a sequence of vertices
\(r_1 \leq r_2 \leq \dots \leq r_\ell\) with the \(r_k\) alternately from the set of \(p_i\) and the
set of \(q_j\) such that \(r_1\) is at a right interval endpoint, \(r_\ell\) is at a left interval
endpoint, and \(r_{k + 1} - r_k = \delta\).
Since \(1 \leq \ell \leq m + n\), this implies that there are only \(O(N^2 \cdot (m + n))\)
candidates for \(\delta\), where \(N\) is the total number of interval endpoints.
We can compute these candidates in time \(O(N^2 \cdot (m + n))\).

Now assume that we have chosen \(\pi\) and \(\sigma\) such that none of the \(p_i\) or \(q_j\) can
be increased (i.e.\@ moved to the right) without increasing the weak Fr\'echet distance.
Then for every \(p_i\) (and likewise \(q_j\)) there is a sequence
\(r_1 \leq r_2 \leq \dots \leq r_\ell = p_i\), where \(r_1\) is the endpoint of an interval and
\(r_{k + 1} - r_k = \delta\).
There are \(O(N)\) possibilities for \(r_1\), \(O(m + n)\) possibilities for \(\ell\), and
\(O(N^2 \cdot (m + n))\) possibilities for \(\delta\), thus the total number of positions to
consider for \(p_i\) is polynomial.

\begin{theorem}
The continuous weak Fr\'echet distance between uncertain one-dimensional curves can be computed in
polynomial time.
\end{theorem}

\subsection{Hardness of Discrete Setting}\label{sec:wdhard}
In this section, we prove that minimising the discrete weak Fr\'echet distance is NP-hard, already
in one-dimensional space.
We show this both in the model where uncertainty regions are discrete point sets and in the model
where they are intervals.

In the constructions in this section, the lower bound Fr\'echet distance is never smaller than~\(1\).
The goal is to determine whether it is equal to~\(1\) or greater than~\(1\).

\subsubsection{Indecisive Points}
We reduce from \textsc{3SAT}.
Consider an instance with \(n\) variables and \(m\) clauses.
We assign each variable a unique \emph{height} (coordinate in the one-dimensional space):
variable \(x_i\) gets assigned height \(10i + 5\).
We use slightly higher heights (\(10i + 6\) and \(10i + 7\)) to interact with a positive state of the
variable, and slightly lower heights to interact with a negative state.

We construct two uncertain curves, one which represents the variables and one which represents the
clauses.
The first curve, \(\U\), consists of \(n + 2\) vertices.
The first and last vertex are certain points, both at height~\(0\).
The remaining vertices are uncertain points, with two possible heights each:
\[\U = \langle 0, \{14, 16\}, \{24, 26\}, \dots, \{10n + 4, 10n + 6\}, 0\rangle\,.\]
The second curve, \(\V\), consists of \(nm + n + m + 2\) vertices. 
For clause \(c_j = \ell_a \lor \ell_b \lor \ell_c\), let \(C_j\) be the set
\(\{10a + 3/7, 10b + 3/7, 10c + 3/7\}\), where for each literal we choose \(+3\) if \(\ell_i = x_i\)
or \(+7\) if \(\ell_i = \neg x_i\).
Let \(S\) be the set \(S = \{15, 25, \dots, 10n + 5\}\) of `neutral' variable heights.
Then \(\V\) is the curve that starts and ends at~\(0\), has a vertex for each \(C_j\), and has
sufficiently many copies of \(S\) between them:
\[\V = \langle 0, S, \dots, S, C_1, S, \dots, S, C_2, S, \dots, S, \dots\dots, C_m, 0\rangle\,.\]

Consider the free-space diagram, with a `spot' \((i, j)\) corresponding to each pair of vertices
\(u_i\) and \(v_j\).
The discrete weak Fr\'echet distance is equal to~\(1\) if and only if there is an assignment to each
uncertain vertex such that the there is a path from the bottom left to the top right of the diagram
that uses only accessible spots, where a spot is accessible if the assigned heights of the
corresponding row and column are within~\(1\).
\Cref{fig:discrete-weak-indecisive} shows an example.

\begin{figure}
\centering
\includegraphics{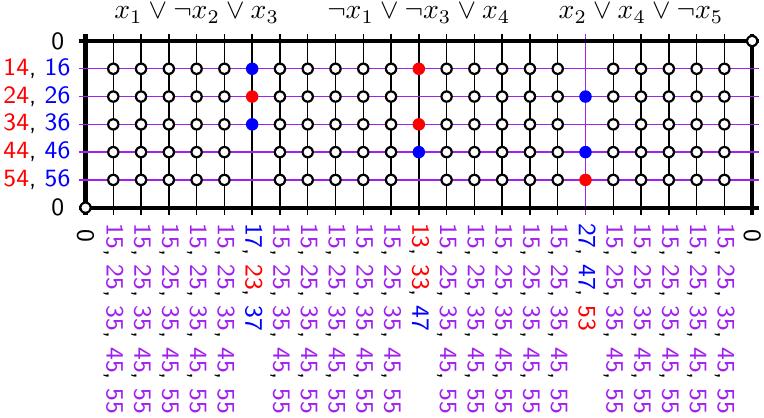}
\caption{An example with five variables and three clauses.
White dots are always accessible, no matter the state of the variables (however, note that only one
white dot per column can be used).
Red / blue dots are accessible only if the corresponding variable is set to \False\ / \True.
Spots without a dot are never accessible.}
\label{fig:discrete-weak-indecisive}
\end{figure}

We can only cross a column corresponding to clause \(c_j\) if at least one of the corresponding
literals is set to true.
The remaining columns can always be crossed at any row.
Note that the repetition is necessary: although all spots are in principle reachable, only one spot
in each column can be reachable at the same time.
If we have at least \(n\) columns between each pair of clauses, this will always be possible.
\begin{theorem}
Given two uncertain curves \(\U\) and \(\V\), each given by a sequence of values and sets of values
in \(\R\), the problem of choosing a realisation of \(\U\) and \(\V\) such that the weak discrete
Fr\'echet distance between \(\U\) and \(\V\) is minimised is NP-hard.
\end{theorem}

\subsubsection{Imprecise Points}
The construction above relies heavily on the ability to select arbitrary sets of values as
uncertainty regions.
We now show that this is not required.
We strengthen the proof in two ways: we restrict the uncertainty regions to be connected intervals,
and we use uncertainty in only one of the curves.

The main idea of the adaptation is to encode clauses not by a single uncertain vertex, but by sets
of globally distinct paths through the free-space diagram.
To facilitate this, we need a global \emph{frame} to guide the possible solution paths, and we need
more copies of the variable vertices (though only one copy will be uncertain) to facilitate the
paths.

Let \(T = 10(n + 2)\).
We build a frame for the construction using four unique heights: \(0\), \(10\), \(T - 10\) and
\(T\).\footnote{The actual values are, in fact, irrelevant for the construction\dsh they simply need
to be unique numbers sufficiently removed from the values we will use for encoding the variables.}
Let \(S = \langle 0, 10, ?, T - 10, ?, 10, ?, T - 10, T \rangle\) be a partial sequence\dsh the
question marks indicate gaps where we insert other vertices later.
Globally, the curves have the structure \(\U = S\) and
\(\V = S \concat S^{-1} \concat S \concat S^{-1} \concat \dots \concat S\): one copy or reversed
copy of \(S\) for each clause (if the number of clauses is even, simply add a trivial clause).
In the free-space diagram, this creates a frame that every path needs to adhere to.
The frame consists of one block per clause, and inside each block, there are three potential paths
from the bottom left to top right corner (or from the top left to bottom right corner for reversed
blocks).
See \cref{fig:discrete-weak-imprecise-frame}.

\begin{figure}
\centering
\includegraphics{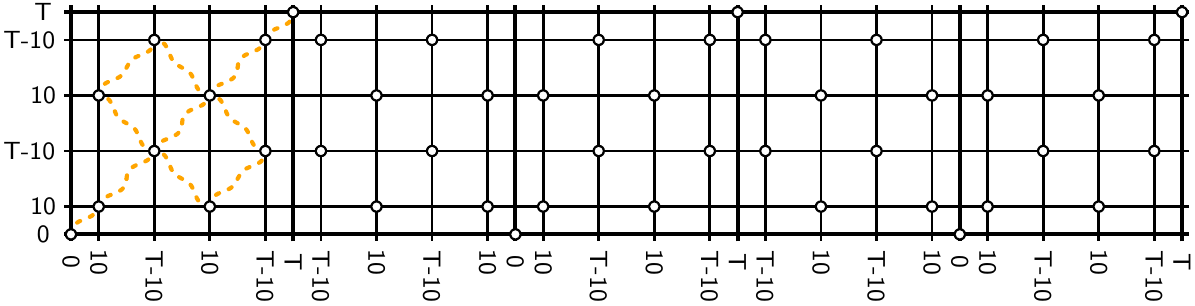}
\caption{The global frame.
White dots are accessible, spots without a dot are never accessible.
Within each block, there are three potential paths between its two accessible corners.}
\label{fig:discrete-weak-imprecise-frame}
\end{figure}

Next, we fill in the gaps.
Let \(\U = \U_1 \concat \U_2^{-1} \concat \U_1\), where
\begin{align*}
\U_1 &= \langle 0, 10, 14, 16, 24, 26, \dots, 10n + 4, 10n + 6, T - 10, T\rangle\,,\\
\U_2 &= \langle 0, 10, [14, 16], [24, 26], \dots, [10n + 4,10n + 6], T - 10, T\rangle\,.
\end{align*}
Let \(\V = \Concat_{1 \leq j \leq m} C_j^{(-1)^{j-1}}\) be concatenation of clause sequences, where
each even clause sequence is reversed.
For a clause \(c_j = \ell_a \lor \ell_b \lor \ell_c\), the sequence \(C_j\) is of the form 
\[C_j = \langle 0, 10\rangle \concat L_a \concat \langle T - 10\rangle \concat L_b^{-1} \concat
\langle 10\rangle \concat L_c \concat \langle T - 10, T\rangle\,,\]
where the literal sequence \(L_i\) corresponding to \(\ell_i = x_i\) (positive literals) or
\(\ell_i = \neg x_i\) (negative literals) is respectively
\begin{align*}
L_i &= \langle 15, 25, \dots, 10(i - 1) + 5, 10i + 5, 10i + 7, 10(i + 1) + 5, \dots, 10n + 5 \rangle\,,\quad\text{or}\\
L_i &= \langle 15, 25, \dots, 10(i - 1) + 5, 10i + 3, 10i + 5, 10(i + 1) + 5, \dots, 10n + 5 \rangle\,.
\end{align*}
See \cref{fig:discrete-weak-imprecise} for an example of the resulting free-space diagram.

\begin{sidewaysfigure}
\centering
\includegraphics{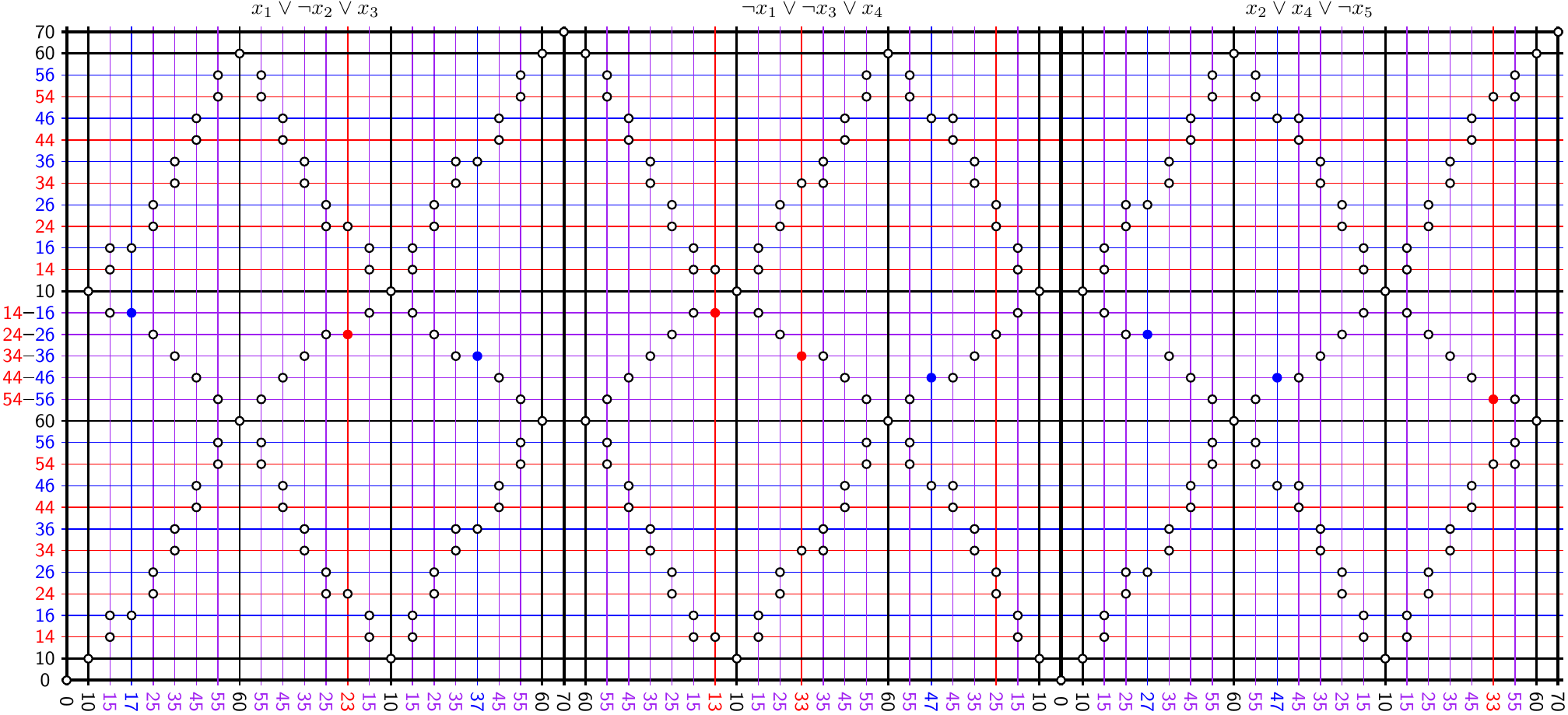}
\caption{An example with five variables and three clauses.
White dots are always accessible, no matter the state of the variables.
Red / blue dots are accessible only if the corresponding variable is set to \False\ / \True.
Spots without a dot are never accessible.}
\label{fig:discrete-weak-imprecise}
\end{sidewaysfigure}

The construction relies on the following.
\begin{observation}
\(L_i\) can always be matched to \(\U_1\).
\(L_i\) can be matched to \(\U_2\) if and only if \(\ell_i = x_i\) and \(x_i\) is set to \True, or
\(\ell_i = \neg x_i\) and \(x_i\) is set to \False.
\end{observation}

\begin{theorem}
Given an uncertain curve \(\U\), given by a sequence of values and intervals in \(\R\), and a
certain curve \(\V\), given by a sequence of values in \(\R\), the problem of choosing a realisation
of \(\U\) such that the weak discrete Fr\'echet distance between \(\U\) and \(\V\) is minimised is
NP-hard.
\end{theorem}

\subsubsection{Continuous Weak Fr\'echet Distance in \texorpdfstring{\(\R^2\)}{R²}}
Finally, we mention that the results in this section carry over to continuous weak Fr\'echet
distance in one dimension higher.
We simply construct the same curves as described above on the \(x\)-axis, and intersperse each curve
with the point at \((0, \infty)\).
\begin{corollary}
Given an uncertain curve \(\U\), given by a sequence of points and regions in \(\R^2\), and a
certain curve \(\V\), given by a sequence of points in \(\R^2\), the problem of choosing a
realisation of \(\U\) such that the weak Fr\'echet distance between \(\U\) and \(\V\) is minimised
is NP-hard.
\end{corollary}

\bibliographystyle{plainurl}
\bibliography{1d_frechet}
\end{document}